\documentclass[journal,comsoc]{IEEEtran}
\usepackage{algorithm,amsbsy,amsmath,amssymb,epsfig,bbm,mathrsfs,multirow,amsthm,xcolor,subcaption, cite}
\usepackage[spaces,hyphens]{url}

\usepackage{hyperref}
\hypersetup{
	colorlinks=true,
	linkcolor=black,
	filecolor=black,   
	urlcolor=black, citecolor=black,}

\usepackage{algorithmic}
\newtheorem{theorem}{Theorem}

\DeclareMathAlphabet\mathbfcal{OMS}{cmsy}{b}{n}
\usepackage{array}

\begin{document}
	%
	\title{AI-enabled mm-Waveform Configuration for Autonomous Vehicles with Integrated Communication and Sensing}
	%
	%
	%
	
	\author{Nam~H.~Chu,
		Diep~N.~Nguyen,
		Dinh~Thai~Hoang,	
		Quoc-Viet Pham, \\
		Khoa~T.~Phan, Won-Joo Hwang, and Eryk~Dutkiewicz
		\thanks{Nam~H.~Chu, Diep~N.~Nguyen, Dinh~Thai~Hoang, and Eryk~Dutkiewicz are with the School of Electrical and Data Engineering, University of Technology Sydney, Australia (e-mails: namhoai.chu@student.uts.edu.au, diep.nguyen@uts.edu.au, hoang.dinh@uts.edu.au, and eryk.dutkiewicz@uts.edu.au).}
		\thanks{Quoc-Viet Pham is with the Korean Southeast Center for the 4th Industrial Revolution Leader Education, Pusan National University, Busan 46241, Korea (e-mail: vietpq@pusan.ac.kr).}
		\thanks{Khoa~T.~Phan is with School of Engineering and Mathematical Sciences, Department of Computer Science and Information Technology, La Trobe University, Melbourne, Australia (e-mail: K.Phan@latrobe.edu.au).}
		\thanks{Won-Joo Hwang is with the Department of Biomedical Convergence Engineering, Pusan National University, Yangsan 50612, Korea (e-mail: wjhwang@pusan.ac.kr).}
	}%

	\maketitle
	
	\begin{abstract}
		Integrated Communications and Sensing (ICS) has recently emerged as an enabling technology for ubiquitous sensing and IoT applications. For ICS application to Autonomous Vehicles (AVs),
		optimizing the waveform structure is one of the most challenging tasks due to strong influences between sensing and data communication functions. 
		Specifically, the preamble of a data communication frame is typically leveraged for the sensing function.
		As such, the higher number of preambles in a Coherent Processing Interval (CPI) is, the greater sensing task's performance is.
		In contrast, communication efficiency is inversely proportional to the number of preambles.
		Moreover, surrounding radio environments are usually dynamic with high uncertainties due to their high mobility, making the ICS's waveform optimization problem even more challenging.
		To that end, this paper develops a novel ICS framework established on the Markov decision process and recent advanced techniques in deep reinforcement learning.
		By doing so, without requiring complete knowledge of the surrounding environment in advance, the ICS-AV can adaptively optimize its waveform structure (i.e., number of frames in the CPI) to maximize sensing and data communication performance under the surrounding environment's dynamic and uncertainty. 		
		Extensive simulations show that our proposed approach can improve the joint communication and sensing performance up to 46.26\% compared with other baseline methods.   	
	\end{abstract}
	

	%
	\IEEEpeerreviewmaketitle

	\section{Introduction}
	\label{sec:intro}

	In Autonomous Vehicles (AVs), e.g., self-driving cars and unmanned aerial vehicles, sensing and data communications are two important functions.
		The sensing function enables AVs to detect objects around them and estimate their distance and velocity for safety management (e.g., collision avoidance).
		The data communication function allows AVs to exchange information with other AVs or infrastructure via Internet of vehicles (IoV). 
		For example, they can send/receive safety messages and even their own raw sensing data (e.g., traffic data around the AV) for applications such as transportation safety, transportation monitoring, and user services distributed to the AVs~\cite{papadimitratos_vehicular_2009}. 
		Although automotive sensing and vehicular communication can share many commonalities (e.g., signal processing algorithms and the system architecture~\cite{zhang_an_2021}) they are typically designed and implemented separately.
		As such, communication and sensing functions require separate hardware components operating at different frequency bands that become increasingly expensive and inefficiency because of an ever-growing number of connected devices and services.
		Consequently, this makes the implementation of communication and sensing functions in AVs more costly in hardware, complexity, and radio spectrum resources.

	These challenges can be effectively addressed by combining both communication and sensing functions into an unified system, called Integrated Communications and Sensing (ICS). 
	This system can utilize an existing communication waveform, such as vehicular communication waveforms, e.g., Cooperative Intelligent Transport Systems (C-ITS)~\cite{naik_ieee_2019} and Dedicated Short-Range Communication (DSRC)~\cite{kenney_dedicated_2011}, WiFi waveform, or cellular waveform, to extract sensing information from targets' echoes.
		Note that there is another type of integrated communication and sensing system, where radar waveforms, e.g., Frequency-Modulated Continuous-Wave (FMCW) waveform, can be used to transfer data~\cite{xie2022waveform}.
		However, it cannot provide a high data rate as required by AVs because the communication signal has to be spread to avoid degrading the sensing performance~\cite{kumari_adaptive_2019}.
		By sharing the same hardware and signals, ICS significantly reduces the power consumption, cost, spectrum usage, and system size compared to conventional approaches where sensing and communication functions are implemented separately, making it more applicable to AVs.
		Hereinafter, an AV with ICS capability is named ICS-AV.

	Currently, two standards operating at $5.9$ GHz for vehicular communication networks are C-ITS based on IEEE 802.11bd in Europe~\cite{naik_ieee_2019} and DSRC based on IEEE 802.11p in the U.S.~\cite{kenney_dedicated_2011}.
		Unfortunately, their data rates (i.e., up to 27 Mbps) do not meet the requirements of AVs' applications. 
		For example, precise navigation that needs to download a high definition three-dimension map and raw sensor data exchange between AVs to support fully automated driving may require connections up to a few Gbps~\cite{choi_millimeter_2016}.
		In addition, the performance of communication and sensing in ICS systems operating at sub-6 GHz is limited due to the bandwidth availability~\cite{kumari_adaptive_2019}. 	
		In this context, millimeter wave (mmWave), whose frequency is from $30$ GHz to $300$ GHz, has been emerging as a promising solution to address the above challenges in ICS systems~\cite{cheng2022integrated}.
		First, owing to the high-resolution sensing and small antenna size, mmWave is predominantly utilized for automotive Long-Range Radar (LRR)~\cite{hasch_milimeter_2012}. 	
		Second, an mmWave system, e.g., a wireless local area network (WLAN) operating at the $60$ GHz band, can provide a very high data rate to meet AVs' intensive communication requirements.

	However, several challenges are hindering the applications of  mmWave ICS systems in AVs.
		In particular, unlike the conventional approaches where sensing and communication are separated, the ICS-AV leverages a single waveform for both sensing and communication functions.
		Thus, it needs to jointly optimize these two functions simultaneously to achieve high performance of data communication and sensing for ICS systems.
		In addition, since the ICS operates while ICS-AVs are moving, the surrounding environments of AVs are highly dynamic and uncertain.
		This makes ICS-AVs' performance unstable as mmWave is more severely impacted by wireless environments than those of the sub-6 GHz bands~\cite{rappaport_wideband_2015}. 
		Therefore, the highly dynamic and uncertainty of mmWave ICS's environment is another critical challenge that needs to be addressed.
		To that end, mmWave ICS systems have been demanding effective and flexible solutions that can not only jointly optimize communication and sensing functions but also adaptively handle the highly dynamic and uncertainty of the surrounding environment to best sustain high data rate communication links (given the highly directional mmWave communications) and sensing accuracy, e.g., low target miss-detection probability and estimation error of target's range and velocity.
	
	A few works in the literature have recently studied communication mmWave waveforms for ICS systems~\cite{kumari_investigating_2015, grossi_oppotunistic_2018, kumari_ieee802_2017, muns_beam_2017, dokhanchi_a_2019, kumari_adaptive_2019}.
		In~\cite{kumari_investigating_2015} and \cite{grossi_oppotunistic_2018}, the authors exploit a single IEEE 802.11ad data communication frame to provide the sensing function. 
		Specifically, the authors in~\cite{kumari_investigating_2015} propose to use the preamble of the Single Carrier Physical Layer (SC-PHY) frame in IEEE 802.11ad to extract sensing information. 
		The simulation results show that this approach can achieve a data rate of up to 1 Gbps with high accuracy in target detection and range estimation.
		However, the velocity estimation is poor because the preamble is short.
		Particularly, the proposed approach achieves the desired velocity accuracy (i.e., 0.1~m/s) only when the Signal-to-Noise Ratio (SNR) is high, i.e., greater than $28$~dB. 
		In~\cite{grossi_oppotunistic_2018}, the authors aim to overcome this issue by using the IEEE 802.11ad Control Physical Layer (C-PHY) frame that has a longer preamble than that of IEEE 802.11ad SC-PHY. 
		However, it is still not large enough to improve the velocity estimation, whereas the data rate is only 27.5 Mbps, significantly lower than the desired data rate for AV's communication~\cite{choi_millimeter_2016}.
		These results from the above studies (i.e., \cite{kumari_investigating_2015} and~\cite{grossi_oppotunistic_2018}) suggest that a single frame processing is unable to satisfy the desired velocity estimation accuracy for AVs. 
		
	Multi-frame processing has been recently considered to be a potential solution for ICS systems to improve sensing information extracted from targets' echoes, e.g., \cite{kumari_ieee802_2017, muns_beam_2017, dokhanchi_a_2019, kumari_adaptive_2019}.
		The authors in~\cite{kumari_ieee802_2017} propose velocity estimation algorithms that leverage multiple fixed-size frames based on IEEE 802.11ad SC-PHY in a CPI.
		Their results demonstrate that the proposed solution can achieve the desired velocity accuracy of AVs (i.e., 0.1 m/s~\cite{hasch_milimeter_2012}) when the number of frames is greater than $20$.
		In \cite{muns_beam_2017}, the authors develop a similar multi-frame processing method to embed sensing functionality in IEEE 802.11ad physical layer frame for a Vehicle-to-Infrastructure (V2I) scenario. 
		By doing so, they can reduce the beam training time of 802.11ad up to 83\%.
		Instead of using the 802.11ad standard, the authors in~\cite{dokhanchi_a_2019} propose a ICS waveform based on Orthogonal Frequency-Division Multiple Access (OFDMA) for a bi-static automotive ICS system, in which the sensing area is extended to non-light-of-sight positions by exploiting reflected signals from other obstacles.
		However, in this work, the maximum communication data rate is only up to $0.1$ Mbps which is dwarfed compared to the desired data rate in AVs.

	A common drawback in the above studies (i.e.,~\cite{kumari_ieee802_2017, muns_beam_2017,dokhanchi_a_2019}) is that the waveform structures (e.g., number of frames in CPI) are not optimized.
		Instead, these parameters are manually set.
		In practice, the dynamic and uncertainty of the ICS's environment (e.g., SNR and data arrival rate) can significantly influence the ICS's data transmission rate as well as sensing accuracy in velocity/range estimation.
		Thus, finding the optimal waveform structure according to the surrounding environment and timely adapting the selected structure with the dynamics of the surrounding environment play vital roles.
		To address this problem, in~\cite{kumari_adaptive_2019}, the authors propose an adaptive virtual waveform design for mmWave ICS based on the 802.11ad standard to achieve the optimal waveform structure (i.e., number of frames in CPI) that can balance between communication and sensing performance.
		The results show that given a fixed length of CPI, increasing the number of frames in CPI can increase the sensing performance, but it will degrade the communication performance (i.e., data rate).	
		However, this approach requires complete information about the surrounding environment in advance, which may be impossible to obtain in practice.
		As such, their proposed solution needs to be rerun from scratch if there is any change in the environment.
	
	In addition, none of the above studies (i.e.,~\cite{kumari_investigating_2015, grossi_oppotunistic_2018, kumari_ieee802_2017, muns_beam_2017, dokhanchi_a_2019, kumari_adaptive_2019}) considers the dynamic and uncertainty problem of the information and environment, e.g., the changes of the wireless channel quality and the arrival rate of data that need to be transmitted via ICS.
		This problem is critical to the performance of the ICS system because the surrounding environment consistently changes as the ICS-AV is moving.
		In particular, the rapid change of the wireless channel quality (e.g., SNR) highly impacts the ICS's communication efficiency (i.e., packet loss due to transmission failure) and sensing performance (i.e., target detection and targets' range and speed estimation accuracy).
		The problem is even more critical for mmWave systems that are highly directional and prone to blockages/fading.
		Moreover, the data arrival process at the ICS-AV is often unknown in advance since it varies in different applications (e.g., navigation and automated driving).
		When the data arrival rate at the AV's ICS system is higher than its maximum transmission rate, data starts to pile up in a data queue/buffer.
		Since a data queue/buffer size is always limited, packet loss will occur when the queue is full.  
		This problem can cause serious issues for AVs as they cannot communicate with other AVs and infrastructure.
		Given above, adaptively optimizing the waveform of ICS is an effective approach to not only jointly optimize both sensing and communication performance but also effectively deal with the dynamic and uncertainty of the surrounding environment.
		However, to the best of our knowledge, this approach has not been investigated in the literature.
		
	To fill this gap, this paper aims to propose a novel framework to maximize the performance of an ICS system by adaptively optimizing the waveform structure under the dynamic and uncertainty of the surrounding environment when the AV is moving.
		It is worth noting that since the sensing processing for the 802.11ad-based ICS is well-investigated in~\cite{kumari_investigating_2015, grossi_oppotunistic_2018, kumari_ieee802_2017, muns_beam_2017, kumari_adaptive_2019}, this study only focuses on addressing the waveform structure optimization problem for mmWave ICS AVs under the dynamic and uncertainty of surrounding environments.
		To that end, we first model the problem as a Markov Decision Process (MDP) because it can allow the AV to determine the optimal waveform (e.g., the number of frames in CPI) based on its current observation (e.g., channel state and number of data packets in the data queue).
		Then, we adopt the Q-learning algorithm, which is widely used in Reinforcement Learning (RL) due to its simplicity and convergence guarantee, to help the ICS-AV gradually learn the optimal policy via interactions with the surrounding environment.
		However, Q-learning may face the curse of dimensionality and overestimation problems that lead to a low converge rate and an unstable learning process when the state space is large~\cite{Hasselt2016_DoubleQ}.
		In our case, the state space that consists of all possible observations of the surrounding environment is very large, while ICS-AV requires fast learning to promptly respond to the highly dynamic and uncertainty of the ICS-AV's environment.  	
	Therefore, we develop a highly-effective learning algorithm based on the most recent advances in RL, namely i-ICS, to deal with these problems.
		First, i-ICS addresses the high dimensional state space problem by utilizing a deep neural network (DNN) to estimate the values of states~\cite{Mnih2015Human}. 
		Second, the overestimation is handled by using the deep double Q-learning~\cite{Hasselt2016_DoubleQ}.	
		Finally, the learning process is further stabilized and accelerated by leveraging the dueling neural network architecture that separately and simultaneously estimates the advantage values and state values~\cite{Wang2016Dueling}. 
		Our major contributions are as follows.
		\begin{itemize}
			\item Design a novel framework by which the ICS-AV can dynamically and automatically optimize the waveform structure under the highly dynamic and uncertainty of its surrounding environment to jointly optimize the  communication efficiency and sensing accuracy, thereby maximizing the ICS's performance.
			\item Develop a highly-effective deep RL (DRL) algorithm taking advantages of recent advances in RL, including deep Q-learning, deep double Q-learning, and dueling neural network architecture, that can help the ICS-AV quickly obtain the optimal policy. 
			\item Perform extensive simulations to investigate the effectiveness of our proposed solution under different scenarios and reveal key factors that can significantly influence the performance of the ICS system.
		\end{itemize}
		The rest of this paper is organized as follows. 
		Sections~\ref{sec:Model} and~\ref{sec:formulation} introduce the ICS system model and the problem formulation, respectively.
		Then, the Q-learning-based and the proposed i-ICS algorithms are proposed in Section~\ref{sec:solutions}.
		In Section~\ref{sec:results}, simulation results are analyzed.
		Finally, we conclude our study in Section~\ref{sec:conclusion}.
	\section{System Model}
	\label{sec:Model}
	
	\begin{figure}
		\centering
		\includegraphics[width=0.95\linewidth]{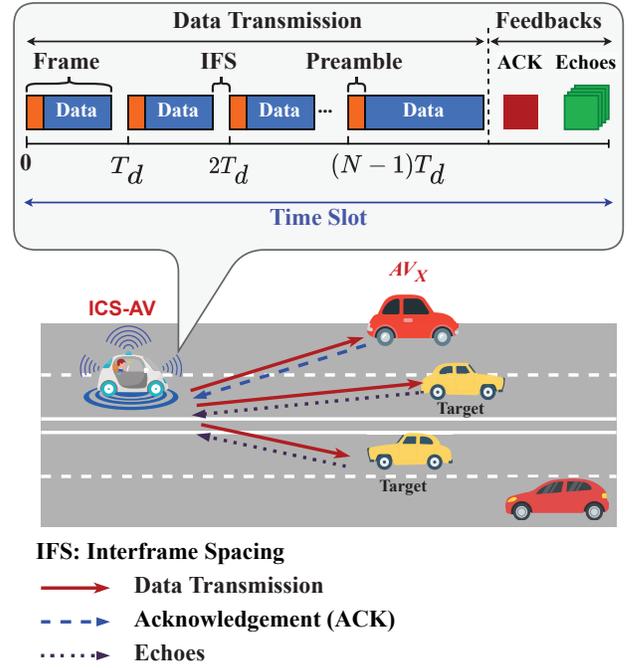}
		\caption{The ICS system model in which the ICS-AV maintains a data communication with $AV_X$ based on IEEE 802.11ad. At the same time, the ICS-AV senses its surrounding environment by utilizing echoes of its transmitted waveforms.}
		\label{fig:system-model}
		\vspace{-10pt}
	\end{figure}
	In this work, we consider an autonomous vehicle, namely ICS-AV, that is equipped with an intelligent millimeter wave integrated communication and sensing (mm-Wave ICS) system based on the IEEE 802.11ad SC-PHY specification.
	Note that this paper leverages the preamble in the SC-PHY for sensing task since it is similar to those in Control Physical Layer (C-PHY) and OFDM Physical Layer (OFDM-PHY), making the proposed solution easily to be extended to other physical layer types in IEEE 802.11ad.
	Moreover, according to~\cite{rohde802.11ad}, the OFDM-PHY is obsolete and may be removed	in a later revision of the IEEE 802.11ad standard.
	The ICS-AV maintains a communication link at a wavelength $\zeta$ with a vehicle $AV_X$, called the recipient vehicle.
	Let $d_X$ and $v_X$ denote the distance and the relative speed between ICS-AV and $AV_X$, respectively.
	At the same time, the ICS-AV gathers echoes of transmitted signals from surrounding targets (e.g., moving vehicles $AV_1, \dots, AV_X$) to perform the sensing function, as depicted in Fig.~\ref{fig:system-model}.	
	In this work, we assume that time is divided into equal slots.
	The time slot is small enough so that the velocities of targets can be considered constant in a time slot~\cite{kumari_ieee802_2017,kumari_adaptive_2019}.	
	At the beginning of a time slot, the ICS-AV decides the mm-Wave ICS waveform structure (i.e., the number of frames in the CPI) that will be used to transmit data in this time slot.
	Then, it observes feedback from the receiving vehicle $AV_X$ (i.e., the acknowledgment of frame) and echoes signals from its surrounding targets at the end of this time slot.
	
	The ICS-AV has a data queue with a maximum of $Q$ packets, each with $B$ Bytes.
	When a new packet arrives, it will be stored in the data queue if this queue is not full. 
	Otherwise, this packet will be dropped. 
	The packet arrival is assumed to follow the Poisson distribution with the mean $\lambda$ packets per time slot.
	Note that the 802.11ad frame has a varying data field; therefore, each frame can contain one or multiple packets.	
	In the following, we first elaborate the proposed ICS transmit and receive signal models, then discuss the sensing processing and ICS performance metrics.
	
	\subsection{Signal Models}
	\subsubsection{Transmitted Signal Model}
	In IEEE 802.11ad, the SC-PHY frame consists of a fixed-size preamble and a varying length data field.
	The preamble contains multiple Golay sequences whose Ambiguity Function (AF) exhibits an ideal auto-correlation without a side-lobe along the zero Doppler axis, making it perfect to be utilized for sensing function, e.g., range estimation and multi-target detection~\cite{ kumari_ieee802_2017}.
	However, its AF is very susceptible to Doppler shifts, leading to a poor velocity estimation.
	To that end, multi-frame processing is proposed to address this problem~\cite{kumari_ieee802_2017, kumari_adaptive_2019}.
	By doing so, the preambles across frames act as radar pulses in a Coherence Processing Interval (CPI).
	This paper considers an ICS waveform structure that consists of $N$ IEEE 802.11ad frames in a CPI.
	The IEEE 802.11ad system can recognize this aggregated frame as the block/no acknowledgment policy~\cite{802.11ad}.
	
	The maximum target's relative velocity, denoted by $v_{max}$, can only be explicitly estimated when frames are located at specific locations in the CPI time~\cite{kumari_adaptive_2019}. 	
	Specifically, the $n$-th frame is located at $nT_d$ (as illustrated in Fig.~\ref{fig:system-model}), where $n \in [0,1,\dots,N-1]$ and $T_d \leq 1/(2\Delta f_{max})$ is a sub Doppler Nyquist sampling interval with a maximum Doppler shift $\Delta f_{max} = 2v_{max}/\zeta$~\cite{richard_fundamentals_2014}.
	Note that the desired ICS system performance can be achieved by optimizing the ICS waveform parameters (e.g., the number of frames in the CPI), which will be described in more details in Section~\ref{sec:model_jrc_metrics}.
	The transmit signal model is then defined as follows. 
	Let $s_n[k]$ denote the symbol sequence corresponding to $n$-th transmitted frame with $K_n$ symbols. 
	Then, the complex-baseband continuous time of transmitted signal in the CPI can be given by~\cite{kumari_adaptive_2019,heath_foundations_2019}:
	\begin{equation}
		\label{eq:transmit_signal}
		x(t) = \sum_{n=0}^{N-1} \sum_{k=0}^{K_n-1} s_n[k]g_{TX}(t-kT_s -nT_d),
	\end{equation}
	where $g_{TX}(t)$ is the unit energy pulse shaping filter at the transmitter of  ICS-AV and $T_s$ is the symbol duration.	
	In this study, similar to ~\cite{kumari_adaptive_2019, va_milimieter_2016}, we consider a single data stream model where the adaptive analog beamforming can be applied to achieve higher directionality beamforming. 
	Thus, the above communication and sensing received signals can be modelled in the following.
	
	\subsubsection{Received Signal Models}
	\paragraph{Data communication received signal}
Suppose that the mmWave communication link between $AV_X$ and ICS-AV is established, the large-scale path loss is defined as follows~\cite{rappaport_wideband_2015}:
		\begin{equation}
			G_c = \frac{G_{TX} G_{RX} \lambda^2}{8\pi^2d_X^2},
		\end{equation}
		where $G_{TX}$ and $G_{RX}$ are the antenna gains of the transmitter (TX) and the receiver (RX), respectively.
		After beamforming, symbol and frequency synchronization phases, the received communication signal is the composition of $P_c$ attenuated and delayed versions of $x(t)$.
		Thus, the received communication signal corresponding to symbol $k$ in the $n$-th frame  can be represented as follows~\cite{kumari_adaptive_2019}:
		\begin{equation}
			\label{eq:comm_signal}
			y^c_n[k] = \sqrt{G_c}\sum_{p=0}^{P_c-1}\beta_c[p] s_n[k-p] + z^c_n[k],
		\end{equation}
		where $z^c_n[k]$ is the complex white Gaussian noise with zero mean and variance
		$\sigma_c^2$, i.e., $\mathcal{N}_C(0,\sigma_c^2)$, and $\beta_c[p]$ is the small-scale complex gain of the $p$-th path.
		Note that we assume that $\beta_c[p]$ is independent and identically distributed (i.i.d) $\mathcal{N}(0,\sigma_p^2)$ where $\sum_{p=0}^{P_c-1}\sigma_p^2=1$, as in~\cite{kumari_adaptive_2019, kumari_ieee802_2017}.
		The SNR of the communication channel is defined as $SNR_c\triangleq E_sG_c/\sigma_c^2$, where $E_s$ is the energy per symbol of the transmitted signal.
	
	\paragraph{Sensing received signal}	
Similar to~\cite{kumari_adaptive_2019, kumari_ieee802_2017, bazzi2012estimation}, this paper uses the scattering center representation to describe the sensing channel.
	We consider that there are $O$ range bins, and at $o$-th range bin there are $B_o$ scattering centers (i.e., targets).
	A scattering center $(d_o, b)$ can be defined by its distance $d_o$, velocity $v_o^b$, radar cross-section $\sigma_o^b$, round-trip delay $\tau_o=d_o/c$ with $c$ being the speed of light, and Doppler shift $\Delta f_o^b = 2v_o^b/\lambda$. 
	The large-scale channel gain corresponding to a scattering center $(d_o, b)$ can be given as follows~\cite{kumari_adaptive_2019, kumari_ieee802_2017, bazzi2012estimation}:
	\begin{equation}
		G_o^b = \frac{G_{TX} G_{RX} \lambda^2\sigma_o^b}{64\pi^3 d_o^4}.
	\end{equation}
	As in\cite{kumari_adaptive_2019, rohling2001waveform}, only a target whose $d_o$ is large in comparison with the distance change during the CPI (i.e., $d_o\geqq v_o^b T_{CPI}$) is considered, so the small-scale channel gain $\beta_o^b$ can be assumed to be constant during the CPI.
	Thus, the received sensing signal model corresponding to symbol $k$ in the $n$-th frame can be given as follows~\cite{kumari_adaptive_2019}:
	\begin{equation}
		y_n[k] = \sum_{o=0}^{O-1} \mathcal{E}_n^o[k] \sum_{b=0}^{B_o-1} \sqrt{G_o^b} e^{-j2\pi\Delta f_o^b(kT_s + nT_d)} + z_n^s[k],
	\end{equation}
	where $z_n^s[k]\sim\mathcal{N}_C(0,\sigma_n^2)$ is the complex white Gaussian noise of the sensing channel, and $\mathcal{E}_n^o[k]$ is the delayed and sampled Matched Filtering (MF) echo from $o$-th range bin, i.e., $\mathcal{E}_n^o[k] = \sum_{i=0}^{K_n-1}s_n[i] g((k-i)T_s - nT_d -\tau_o)$, where $g(t) = g_{TX}(t) * g_{RX}(t)$ is the net TX-RX pulse shaping filter.
	
	In this study, similar to~\cite{kumari_virtual_2018, kumari_adaptive_2019, rohling2001waveform}, we assume that the channel is stationary during the preamble period of a frame due to the small preamble duration.
	As such, the received signal model corresponding to the preamble $\mathcal{E}_t^o[k]$ of a frame can be given as follows~\cite{kumari_adaptive_2019}:
	\begin{equation}
		\label{eq.recieved_preamble}
		y_n^t[k] = \sum_{o=0}^{O-1} \mathcal{E}_t^o[k] \sum_{b=0}^{B_o-1} \sqrt{G_o^b} e^{-j2\pi\Delta f_o^bnT_d} + z_n^s[k].
	\end{equation}
	Note that $kT_s$ can be omitted from the phase shift term in the signal model corresponding to the preamble part since the channel is assumed to be time-invariant within the preamble period.

	\subsection{Sensing Signal Processing}
	We now discuss sensing signal processing in the ICS system.
	Based on the cross-correlation output between the transmitted preamble of a single 802.11ad frame and the received signal, the ICS system can detect a target with high probability (more than 99.99\%) and achieve the desired range resolution (i.e., $0.1$ m~\cite{hasch_milimeter_2012}) for automotive LRR~\cite{kumari_ieee802_2017}.
	However, the AF of IEEE 802.11ad preamble is sensitive to Doppler shift, making it less accurate in velocity estimation.
	Therefore, this work only considers the velocity estimation of the ICS system, which is more challenging to obtain a high accuracy than those of target detection and range estimation processes. 
	
	After detecting targets and obtaining the corresponding range bins, the velocity estimation can be executed as follows.
		Given the $n$-th frame received in~\eqref{eq.recieved_preamble}, the sensing channel corresponding to detected targets at the $o$-th range bin can be expressed as follows~\cite{kumari_adaptive_2019}:
		\begin{equation}
			h_o^n= \sum_{b=0}^{B_o-1}u_o^b e^{-j2\pi\Delta f_o^bnT_d} + z_o^n,
		\end{equation}
		where \mbox{$u_o^b \!=\! \gamma \sqrt{E_sG_o^b}$} is the signal amplitude, $\gamma$ is the correlation integration gain, and $z_o^n$ is the complex white Gaussian noise $\mathcal{N}_C(0,\sigma_n^2)$.
		Then, the channel vector corresponding to the $o$-th range bin for $N$ frames in the CPI, i.e., \mbox{$\mathbf{h}_o \!=\! [h_o^1, h_o^2, \dots, h_o^{N-1}]$}, can be given as follows:
		\begin{equation}
			\label{eq:channel_vector}
			\mathbf{h}_o = \mathbf{D}_o\mathbf{u}_o + \mathbf{z}_o,
		\end{equation}
		where \mbox{$\mathbf{z}_o \!=\! [z_o^0, z_o^1, \dots, z_o^{N-1}]$} is the channel noise vector, \mbox{$\mathbf{u}_o \!\triangleq\! [ u_o^0, u_o^1, \dots, u_o^{B_o-1} ]^T$} denotes the channel signal amplitude vector, \mbox{$\mathbf{d}(v_o^b) \!\triangleq\![1, e^{-j2\pi\Delta f_o^bT_d}, \dots ,\allowbreak, e^{-j2\pi\Delta f_o^b(N-1)T_d} ]^T$} is the vector of channel Doppler corresponding to $b$-th velocity at $o$-th range bin, and $\mathbf{D}_o \!\triangleq\![\mathbf{d}(v_o^0), \mathbf{d}(v_o^1), \dots, \mathbf{d}(v_o^{B_o-1}) ]$ is the matrix of channel Doppler.
		The target velocity can be estimated based on \eqref{eq:channel_vector} by using Fast Fourier transform (FFT)-based algorithms that are widely used in the classical radar processing~\cite{kumari_ieee802_2017, richard_fundamentals_2014}.
	
	\subsection{ICS Performance Metrics}
	\label{sec:model_jrc_metrics}
	As discussed in the previous subsection, this work focuses on the target velocity estimation.
	Thus, the sensing performance for the ICS can be determined by the velocity estimation accuracy (i.e., velocity resolution).
	For the FFT-based velocity estimation approach, the velocity resolution is defined by~\cite{kumari_ieee802_2017}:
	\begin{equation}
		\label{eq:velocity_resolution}
		\Delta_v = \frac{\zeta}{2N_fT_d},
	\end{equation}
	where $N_f$ is the number of frames used for the velocity estimation process.
	Then, the velocity measurement accuracy can be characterized by the root mean square error that depends on the SNR of the received sensing signal as follows~\cite{hasch_milimeter_2012, curry_radar_2004}:
	\begin{equation}
		\label{eq:velocity_accuracy}
		\delta = \frac{\zeta}{2N_fT_d \sqrt{2\mbox{\textit{SNR}}_r}}.
	\end{equation}
	Equation \eqref{eq:velocity_accuracy} implies that given a fixed CPI time, a fixed $T_d$, and a constant data rate, the velocity estimation accuracy increases (i.e., $\delta$ decreases) as the value of $N_f$ increases. 
	Recall that frames are placed at consecutive multiples of $T_d$ in the CPI. 
	Therefore, the last frame's size is larger than those of others if the total number of frames is less than the maximum number of frames in the CPI.   
	As such, as the number of frames in the CPI increases, the size of the last frame decreases since the CPI duration is constant.
	
	The communication metrics for the ICS must represent the data transmission performance.
	Two typical communication metrics are the transmission rate and reliability (e.g., packet loss).
	Thus, this work considers two metrics, 1) the length of the data queue representing the efficiency of data transmission, and 2) the number of dropped packets demonstrating the system reliability.
	
	It is also worth noting that a large frame has a higher drop probability than that of a smaller frame in the same wireless environment.
	Thus, using multiple small-size frames can increase not only the reliability of transmission but also the sensing performance (i.e., the velocity estimation accuracy).
	However, it leads to an overhead because it increases the number of preambles that do not contain user data.
	Consequently, packets pile up at the ICS-AV, leading to packet drop in the queue.	
	In addition, the characteristics of the ICS-AV's surrounding environment (e.g., the wireless channel quality, which can be represented through the packet drop probability and SNR, and the packet arrival rate) highly influence the ICS performance in regard to data transmission reliability and sensing accuracy.	
	Therefore, the ICS system needs to obtain an optimal policy that optimizes the ICS waveform (i.e., number of frames in the CPI) to achieve the desired performance in terms of data transmission and sensing accuracy.
	Furthermore, the ICS-AV's surrounding environment may change significantly from time to time, especially in the ICS environment where AVs usually travel.
	Thus, optimizing the ICS waveform in each time slot is an intractable problem. 
	The following sections will describe our proposed MDP framework for the ICS operation problem that enables the ICS-AV to quickly and effectively learn the optimal policy without requiring complete information from the surrounding environment, thereby achieving the best performance compared with traditional solutions.  
	
\section{Problem Formulation}
	\label{sec:formulation}
	The ICS-AV's operation problem is formulated using the MDP framework to deal with the highly dynamic and uncertain of the surrounding environment.
	The MDP framework can help the ICS-AV adaptively decide the best action (i.e., the ICS waveform structure) based on its current observations (i.e.,  the current data queue length and channel quality) at each time slot to maximize the ICS system performance without requiring complete knowledge on the packet arrival rate, data transmission process, and channel quality in advance.  
	An MDP is generally defined by the state space $\mathcal{S}$, the action space $\mathcal{A}$, and the immediate reward function $r$.
	The following sections will discuss more details about the components of our proposed framework. 
	\subsection{State Space}
	As we aim to maximize the performance of the ICS system with regard to data transmission efficiency and sensing accuracy, we need to consider the following key factors. 
	The first one is the current data queue length (i.e., the number of packets in the data queue) because it reflects the efficiency of the data transmission process. 
	For example, given a packet arrival rate, the lower the number of packets in the queue is, the higher the data transmission efficiency is.
	The second one is the link quality that can be estimated by using the SNR metric at the ICS-AV.
	At the beginning of a time slot, the link quality is estimated based on the feedback (i.e., the recipient vehicle's ACK frame and targets' echoes) of the transmitted frame in the previous time slot.
	Although it does not represent the instantaneous channel state, it help to provide valuable information about the surrounding environment.
	
	In this work, we consider that the channel quality level can be grouped into $C$ different classes, which are analogous to the Modulation and Coding Scheme (MCS) levels in IEEE 802.11ad~\cite{802.11ad}. 
	These classes have different probabilities of bit errors due to the different wireless channel qualities, denoted by a probability vector $\mathbf{p_e}=[p_1,p_2,\dots,p_C]$.
	Note that given the transmission link with bit error probability $p_b$, the error probability of an $F$-bit frame can be calculated by $p_f = 1-(1-p_b)^F$~\cite{khalili_a_new_2005}.
	In addition, if a frame drops, all packets in this frame will be lost.		 
	To that end, the ICS's state space can be given as follows:
	\begin{equation}
		\mathcal{S} = \Big\{(q, c): q \in \{0, \dots, Q \}; c \in \{0, \dots, C\}  \Big\},
	\end{equation}
	where $q$ is the current number of packets in the data queue and $c$ is the channel quality. Here, $Q$ is the maximum number of packets that the data queue can store.
	In this way, the system state can be represented by a tuple $s = (q,c)$.   
	By this design, the ICS system continuously operates without falling into the terminal state.
	
	\subsection{Action Space}
	As discussed in Section~\ref{sec:model_jrc_metrics}, the  ICS waveform plays a critical role in the system performance.
	In particular, given a fixed CPI time $T$, using a large number of frames results in a high reliability of data transmission and a high sensing accuracy.
	However, it reduces the efficiency of data transfer as there are more overhead data. 
	At each time slot, the  ICS-AV needs to select the most suitable ICS waveform structure (i.e., the number of frames in the CPI) to maximize the system performance.
	Thus, the action space can be defined as follows:
	\begin{equation}
		\mathcal{A} = \big\{1, \dots, N \big\},		
	\end{equation}  
	where $N$ is the maximum number of frame in the CPI. 
	Recall that the beginning of each frame needs to be placed at the multiple of $T_d$ consecutively, and thus $N = \lfloor \frac{T_{CPI}}{T_d} \rfloor$, where $T_{CPI}$ is the CPI time and $\lfloor \cdot \rfloor$ is the floor function.
	As such, if the number of frames in the CPI selected by the ICS-AV is less than $N$, the last frame will be longer than others.
	Note that when the data queue is empty, the  ICS-AV can still send dummy frames (e.g., frames whose data fields contain random bits) to maintain the ICS's sensing function continuously. 
	
	\subsection{Reward Function}	
	Since the ICS system performs two functions simultaneously, i.e., data transmission and sensing, we aim to maximize the ICS system performance by balancing the data transmission efficiency and the sensing accuracy.
	Thus, the reward function needs to capture both of them. 
	The data transmission efficiency can be defined according to the number of packets waiting in the queue and the number of dropped packets.
	Specifically, the lower the number of packets in the data queue and the number of dropped packets are, the higher the efficiency of the ICS system is.
	Suppose that at time slot $t$, the ICS-AV observes state $s_t$ and takes action $a_t$.
	Let $q_t$, $\delta_t$, and $l_t$ denote the current size of the data queue, the sensing accuracy, and the number of dropped packets that the ICS-AV observes at the end of $t$, respectively.
	Then, the immediate reward function can be defined as follows:
	\begin{equation}
		\label{eq:reward_function}
		r_t (s_t, a_t)= -(w_1q_t + w_2\delta_t + w_3l_t),
	\end{equation}
	where $w_1$, $w_2$, and $w_3$ are the weights to tradeoff between the number of packets waiting in the queue, the sensing accuracy, and the number of dropped packets due to the data queue full.
	The negative function in~\eqref{eq:reward_function} implies that the ICS-AV should take an action that can quickly free the data queue, lower the number of dropped packets, and achieve a high sensing accuracy.
	Note that the lower the value of $\delta_t$ is, the higher the velocity estimation accuracy of the system is. 
	{Given the above, the immediate reward function~\eqref{eq:reward_function} effectively captures joint performance of communication and sensing function of the ICS.}
	\subsection{Optimization Formulation}
	The objective of this study is to find an optimal policy for the ICS-AV that maximizes the long-term reward function.
	Let $R(\pi)$ denote the long-term average reward function under policy $\pi: \mathcal{S} \rightarrow \mathcal{A}$, then the problem can be formulated as:
	\begin{eqnarray} 
		\label{eq:average_reward}
		\max_\pi \>	{R}(\pi)	=	\lim_{T \rightarrow \infty} \frac{1}{T} \sum_{t=1}^{T} {\mathbb{E}} \big( r_t(s_t,\pi(s_t)) \big),		
	\end{eqnarray}  
	where $\pi(s_t)$ is the action at time $t$ according to policy $\pi$.
	Thus, given the ICS-AV's current data queue length and wireless channel quality, the optimal policy $\pi^*$ gives an optimal action that maximizes $R(\pi)$.
	In addition, Theorem~\ref{theorem1} shows that the average reward function is well defined regardless of the initial state.
	
	\begin{theorem}
		\label{theorem1}
		With the proposed MDP framework, the average reward function $R(\pi)$ is well defined under any policy $\pi$ and regardless of a starting state.
	\end{theorem}

	\begin{proof}
		We first prove that the Markov chain of the considered problem is irreducible as follows.
		Recall that the state of the ICS consists of two factors, i.e., the current queue length $q$ and the wireless channel quality $c$. 
		For each time slot, the data arrival rate is assumed to follow the Poison distribution and the channel quality is derived from $C$ class accordingly a probability vector $\mathbf{p_e}=[p_1,p_2,\dots,p_C]$.
		Therefore, given the ICS is at state $s$ at time $t$, it can move to any other states $s' \in \mathcal{S}\{s\}$ after finite time steps.
		As such, the proposed MDP is irreducible with the state space $\mathcal{S}$, thereby making the average reward function $R(\pi)$ is well defined under any policy $\pi$ and regardless of a starting state.		  
	\end{proof}
	
	\section{Reinforcement Learning based Solution for ICS-AV Operation Policy}	
	\label{sec:solutions}
	Due to the highly dynamic and uncertainty of the environment (e.g., packet drop probability due to the channel quality and the data arrival rate), the ICS-AV is unable to obtain this information in advance.
	In this context, RL can help the ICS-AV obtain the optimal policy without requiring completed knowledge about surrounding environment in advance.
	The idea is that the ICS-AV can gradually learn through interacting with its surrounding environment.	
	In the following, we first present a Q-learning based approach that is one of the most popular algorithms in RL due to its simplicity and convergence guarantee.
	We then develop an intelligent solution that can effectively overcome limitations of Q-learning by adopting three recent advanced techniques in RL, namely deep Q-learning, deep double Q-learning, and dueling neural network architecture. 
	
	\subsection{Q-learning based ICS's Waveform Structure Optimization}
	\label{subsec:Q-learning}
	In RL, the value of state $s$ at any time step $t$ under policy $\pi:\mathcal{S}\rightarrow\mathcal{A}$ is calculated by the state value function, i.e.,
	\begin{equation}
		v_{\pi}(s) = \mathbb{E}_{\pi} \bigg[\sum_{m=0}^{\infty}\eta^m r_{t+m}\Big|s_t=s \bigg], \forall s \in \mathcal{S},
	\end{equation}
	where $\mathbb{E}_{\pi}[\cdot]$ is the expectation under policy $\pi$ and $\eta$ is the discount factor that indicates the importance of future rewards.
	Similarly, the state-action function evaluates how good to performing action $a$ at state $s$ then following policy $\pi$, which is given by:
	\begin{equation}
		q_{\pi}(s,a) = \mathbb{E}_{\pi} \bigg[\sum_{m=0}^{\infty}\eta^m r_{t+m}\Big|s_t=s, a_t=a \bigg],
	\end{equation}   
	Under the optimal policy $\pi^*$, the optimal state-action value function, i.e., $q^*(s,a)$, is given by~\cite{sutton_intro_1992}:
	\begin{equation}
		q^*(s,a) = \mathbb{E}\Big[r_t + \eta \> \max_{a'\in\mathcal{A}}\> q^*(s_{t+1},a')\big|s_t=s, a_t=a\Big].
	\end{equation}
	Thus, once $q^*(s,a)$ is obtained, the optimal policy is achieved by taking actions that maximize $q^*(s,a)$ for all state $s\in \mathcal{S}$. 
	
	Q-learning algorithm uses a table, namely Q-table, to learn $q^*(s,a)$, making its implementation simple. 
	The Q-learning based approach for the ICS-AV is presented in details in Algorithm~\ref{alg:qlearning}. 
	Specifically, each cell of the Q-table keeps the estimated value of Q-function (named Q-value) for taking action $a$ at state $s$, denoted by $q(s,a)$. 
	The Q-table is iteratively updated based on interactions with the surrounding environment.
	Given action $a_t$ is selected under the $\epsilon-$policy~\eqref{eq:epsilon_greedy} at state $s_t$ at time $t$, the ICS-AV obtains the immediate reward $r_t$ and moves to a next state $s_{t+1}$.
	Then, $q(s_t,a_t)$ is updated as follows:
	\begin{equation}
		\label{eq:updateQfunction}
		\begin{aligned}
			q(s_t,a_t) \leftarrow &q(s_t,a_t) + \\ &\alpha_t\big[ \underbrace{\underbrace{r_t(s_t, a_t) + \eta\max_{a_{t+1}} q(s_{t+1}, a_{t+1})}_{\text{Target Q-value } Y_t} - q(s_t,a_t)}_{\text{Temporal difference (TD)}}\big],
		\end{aligned}
	\end{equation}
	where $\alpha_t$ is the learning rate that controls how important of new knowledge (i.e., TD) to the update of Q-function.
	By using~\eqref{eq:updateQfunction} and $\alpha_t$ that satisfies~\eqref{eq:learning_rate_rules} to iteratively updating the Q-function, it is proven that $q(s,a)$ will converge to $q^*(s, a)$~\cite{Watkins1992QLearning}. 
	\begin{equation}
		\label{eq:learning_rate_rules}
		\alpha_t \in [0,1), ~\sum_{t=1}^{\infty}\alpha_t = \infty, \mbox{ and } \sum_{t=1}^{\infty} ( \alpha_t  )^{2} < \infty.
	\end{equation}
	\begin{algorithm}[t]
		\caption{Q-learning based ICS waveform optimization}
		\label{alg:qlearning}
		\begin{algorithmic}[1]
			\STATE The ICS-AV establishes the parameters (i.e., $\eta$, $\alpha$, and $\epsilon$) and create a Q-table arbitrarily (e.g., all cells are set to zero)
			\FOR{\textit{t = 1 to T}}
			\STATE The ICS-AV performs an action following the $\epsilon$-greedy policy as follows:
			\begin{align}
				a_t\! =\! \left\{
				\begin{array}{ll}						
					\!\underset{a \in \mathcal{A}}{\text{argmax}}~Q(s_t,a),&\mbox{with probability }1-\epsilon,\\
					\!\mbox{random action }a\!\in\!\mathcal{A}, &\mbox{otherwise}.												
				\end{array}	\right.
				\label{eq:epsilon_greedy}
			\end{align}
			\STATE ICS-AV observes next state $s_{t+1}$ and reward $r_t$, then updates $Q(s_t,a_t)$ and reduces $\epsilon$ by~\eqref{eq:updateQfunction}
			\ENDFOR
		\end{algorithmic}
	\end{algorithm}
	However, the usage of a table in Q-learning leads to the curse of dimensionality problem that results in a long time for learning, especially in the considered ICS system with a high dimensional state space.
	In addition, the uncertainty and dynamic of the ICS's environment (e.g., the probability of frame loss and the packet arrival rate) make it more challenging for the ICS-AV to achieve an optimal policy.
	Furthermore, Q-learning is prone to overestimate the Q-values due to the max operation in~\eqref{eq:updateQfunction}, i.e., the Q-value of a state-action pair is estimated higher than its actual value~\cite{Hasselt2016_DoubleQ}.
	It may not be crucial if all Q-values are uniformly shifted. 
	Nevertheless, they are typically non-uniform in practice, thereby substantially slowing down the ICS-AV's learning process.	
	To that end, we will discuss our highly-effective intelligent algorithm (namely i-ICS) to quickly achieve the optimal operation policy for the ICS-AV, thus maximizing the data transmission efficiency and sensing accuracy of the system.
	\subsection{Deep Reinforcement Learning based ICS's Waveform Structure Optimization}
	Instead of using a Q-table, a DNN can be employed to quickly learn the optimal Q-function and address the high dimensional state and action spaces effectively ~\cite{Mnih2015Human}.
	The employment of DNN in RL forms a new group of RL algorithms, i.e., DRL.
	Thus, this paper develops a DRL algorithm, i.e., i-ICS, that can efficiently overcome the slow convergence and overestimation problems of Q-learning by adopting three innovation techniques in RL, including 1) Deep Q-learning (DQ)~\cite{Mnih2015Human}, 2) Double Deep Q-learning (DDQ)~\cite{Hasselt2016_DoubleQ}, and 3) dueling architecture~\cite{Wang2016Dueling}.
	In this way, our proposed approach can inherit all advantages of these techniques, thereby stabilizing the learning process, improving the learning speed, and reducing the overestimation.
	As a result, the ICS-AV can quickly obtain the optimal policy to maximize both data transmission efficiency and sensing accuracy.
	\begin{algorithm}[t]
		\caption{The i-ICS Algorithm}
		\label{alg:drl-jcr}
		\begin{algorithmic}[1]
			\STATE Initialize buffer $\mathbf{E}$ and $\epsilon$. 
			\STATE Create Q-network $Q$ with arbitrary parameters $\theta$ 
			\STATE Create target Q-network $\hat{Q}$ with parameters $\theta^-=\theta$.
			\FOR{\textit{t = 1 to T}}
			\STATE Based on $\epsilon$-greedy policy~\eqref{eq:epsilon_greedy}, select action $a_t$.
			\STATE Perform $a_t$, observe next state $s_{t+1}$ and receive reward $r_t$.
			\STATE Store $(s_t, a_t, r_t, s_{t+1})$ in $\mathbf{E}$.
			\STATE Create an experience mini-batch by sampling from $\mathbf{E}$ randomly, denoted by $(s, a, r, s') \sim U( \mathbf{E}$).				
			\STATE Compute $Q(s_t, a_t; \theta)$ and $Y_t$ by~\eqref{eq:recontruct_Qfunction} and~\eqref{eq:targetDDQN}, respectively.			   
			\STATE Update the Q-network's parameters by SGD.
			\STATE Reduce $\epsilon$.
			\STATE Replace  $\hat{Q} \leftarrow Q$ at every $U$ steps.
			\ENDFOR
		\end{algorithmic}
	\end{algorithm}

	The details of i-ICS are presented in Algorithm~\ref{alg:drl-jcr}.
	In particular, the major steps are an analogy to those of Q-learning.
	However, the experience at time $t$ is not used directly to train the neural network, but it is stored in a buffer $\mathbf{E}$ instead.
	The reason is that the successive experiences are highly correlated, resulting in severely slowing the learning process ~\cite{halkjaer_the_1996}.
	Then, a mini-batch of experiences is sampled uniformly at random from $\mathbf{E}$ to train the neural network.	 
	By doing so, the high correlation between successive experiences decreases so that the convergence is accelerated. 
	In addition, the usage of the experience buffer increases the data usage efficiency since each experience can be utilized multiple times to train the DNN.
	In i-ICS, the optimal Q-function is estimated by a DNN where the input layer consists of two inputs corresponding to the dimension of ICS-AV's state (i.e., the data queue size and the channel quality), whereas the number of neurons in the output layer is the maximum number of frames in the CPI, i.e., $N$. 
	Specifically, given that state $s$ is passed to the DNN, each output layer's neuron value corresponds to a Q-value of one action in this state.
	\begin{figure}[t]
		\centering
		\includegraphics[width=0.99\linewidth]{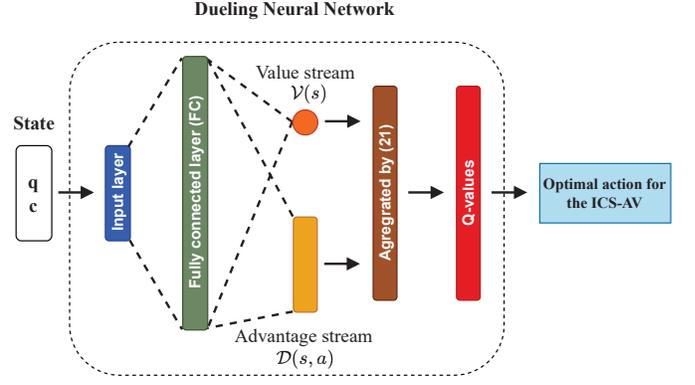}
		\caption{The dueling architecture for i-ICS in which a DNN is divided into two streams, one for estimating the state value function $v(s)$ and another for estimating the advantage function $d(s,a)$.}
		\label{fig:dueling-architecture}
		\vspace{-10pt}
	\end{figure}
	
	To stabilize the learning process, we leverage a dueling neural network architecture where the DNN is split into two streams, one for estimating the state value function $v(s)$ and another for estimating the advantage function $d(s,a)$, as shown in Fig.~\ref{fig:dueling-architecture}. 	
	Recall that $q(s,a)$ evaluates the value of executing an action at a state, whereas $v(s)$ describes how good to end up at a state. 
	Then, $d(s,a)$ is defined by the relationship between $v(s)$ and $q(s,a)$ under a policy $\pi$ as:
	\begin{equation}
		\label{eq:advantage_function}
		d^\pi(s,a) = q^\pi(s,a) - v^\pi(s).
	\end{equation}
	Thus, $d(s,a)$ demonstrates the importance of action $a$ compared to other actions at state $s$.
	Let $\phi$ and $\psi$ denote the state value and advantage streams' parameters of the dueling DNN, respectively.
	Then, we denote $Q(s,a; \phi,\psi)$, $V(s; \phi)$, and $D(s,a; \psi)$ as the estimations of $q(s,a)$, $v(s)$, and $d(s,a)$ given by the dueling network, respectively.
	Hence, the value of taking action $a$ at state $s$ given by the dueling DNN is expressed as:
	\begin{equation}
		\label{eq:recontruct_Qfunction}
		Q(s,a; \phi,\psi) = V(s; \phi) + D(s,a; \psi).		
	\end{equation}
	Note that \eqref{eq:recontruct_Qfunction} is unidentifiable, i.e., for a given Q-value, it is unable to uniquely determine $V(s; \phi)$ and $D(s,a; \psi)$ since the Q-value is unchanged if $V(s; \phi)$ decreases the same amount that $D(s,a; \psi)$ increases.
	Therefore, a learning algorithm may have poor performance if it uses \eqref{eq:recontruct_Qfunction} directly.
	This issue can be addressed as follows~\eqref{eq:recontruct_Qfunction_max}.
	\begin{equation}
		\label{eq:recontruct_Qfunction_max}
		Q(s,a; \phi, \psi) = V(s; \phi) + \Big(D(s,a; \psi) - \max_{a' \in \mathcal{A}} D(s,a';\psi)\Big).		
	\end{equation}
	By doing so, the estimation of the advantage function is forced to zero at the selected action. 
	Specifically, given $a^* \!=\!  \text{argmax}_{a \in \mathcal{A}}\>Q(s,a; \phi,\psi) \!=\! \text{argmax}_{a \in \mathcal{A}}\>D(s,a; \psi)$, we have $Q(s,a^*; \phi,\psi)=V(s; \phi)$.
	Thus, $V(s; \phi)$ gives an estimation of the state value function, while $D(s,a';\psi)$ provides an estimation of the advantage function~\cite{Wang2016Dueling}.	
	Nevertheless, because the change of the advantage stream is as fast as that of the predicted optimal action's advantage, the estimation of Q-function may be unstable.
	To overcome this problem, the max operation in~\eqref{eq:recontruct_Qfunction_max} can be replaced by the mean as follows~\cite{Wang2016Dueling}:   
	\begin{equation}
		\label{eq:recontruct_Qfunction_mean}
		Q(s,a; \phi,\psi) \!=\! V(s; \phi)\! +\! \Big(D(s,a; \psi)\! -\! \frac{1}{|\mathcal{A}|}\! \sum_{a'}\!D(s,a';\psi)\Big).		
	\end{equation}
	\begin{figure}[t]
		\centering
		\includegraphics[width=0.99\linewidth]{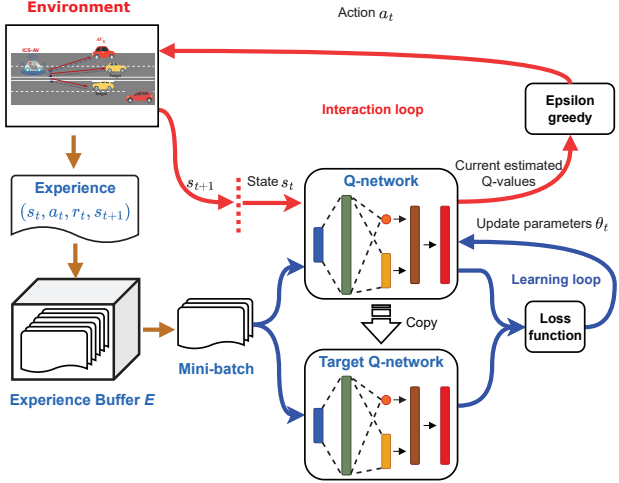}
		\caption{The proposed i-ICS model, in which the ICS-AV obtains an optimal policy by gradually updating its policy based on its observations of the surrounding environment.}
		\label{fig:DRL-ICS_model}
		\vspace{-10pt}
	\end{figure}
	The proposed i-ICS addresses the Q-learning's overestimation by adopting the double deep Q-learning~\cite{Hasselt2016_DoubleQ}.
	In particular, i-ICS utilizes two dueling neural networks whose architectures are identical: 1) Q-network $Q$ for selecting an action and 2) target Q-network $\hat{Q}$ for evaluating an action.
	Hereinafter, let  $\theta_t$ denote the Q-network's parameters and $\theta_t^-$ denote the target Q-network's parameters at time $t$.
	Thus, the target Q-value of performing action $a$ at state $s$ at time $t$ is given by:
	\begin{equation}
		\label{eq:targetDDQN}
		Y_t = r_t + \eta~\hat{Q}\big(s_{t+1}, \underset{a'}{\operatorname{argmax}}~Q(s_{t+1},a'; \theta_t);\theta^-_t\big).
	\end{equation}
	Since minimizing the TD is the purpose of training the Q-network, the loss function can be given as:
	\begin{equation}
		\begin{aligned}
			\label{eq:lossfunction}
			L_t(\theta_t) = \mathbb{E}_{(s,a,r,s')}\bigg[ \bigg( Y_t
			- Q(s,a;\theta_t)\bigg)^2\bigg],
		\end{aligned}
	\end{equation} 
	where $(s,a,r,s')$ is a data point (i.e., an experience) in memory $\mathbf{E}$.

	In deep learning, Gradient Descent (GD) is widely employed to minimize the loss function since it can obtain global minima and is straightforward to be implemented~\cite{du_gradient_2019}.
	At time $t$, the cost function in GD is computed as follows:
	\begin{equation}
		\begin{aligned}
			\label{eq:costfunction}
			C_t(\theta_t) = \frac{1}{|\mathbf{E}|}\sum_{(s,a,r,s') \in \mathbf{E}}L_t(\theta_t).
		\end{aligned}
	\end{equation}
	Then, GD updates neural network parameters as: $\theta_{t+1} = \theta_t -\beta_t \nabla_{\theta_t} C_t(\theta_t),$
	where $\beta_t$ is a step size at time $t$ and $\nabla_{\theta_t}[.]$ is the gradient function with respect to parameters $\theta_t$.
	The main shortcoming of GD is that it requires all data points in $\mathbf{E}$ to compute the loss and the gradient of the cost function for each update of parameters, making it sluggish when the data pool is enormous.
	Therefore, we propose to use stochastic gradient descent (SGD) to minimize the loss function in~\eqref{eq:costfunction} since it can accelerate the learning convergence~\cite{Robbins1951SGD}.
	In particular, for each iteration, SGD only computes the cost and its gradient for a mini-batch of experiences that is uniformly sampled from $\mathbf{M}$, making its computational complexity much lower than that of SGD.
	Note that in~\eqref{eq:lossfunction}, the target Q-value $Y_t$ looks like labels in supervised deep learning, which are fixed before the training process.
	Nevertheless, $Y_t$ will change if there is any change in the target Q-network's parameters~$\theta_t^-$, making the training process unstable.
	To that end, $\theta_t^-$ is only updated by copying from $\theta_t$ at every $U$ steps, as shown in Fig.~\ref{fig:DRL-ICS_model}. 

\subsection{Computational Complexity Analysis}
	The computational complexity of our proposed algorithm, i.e., i-ICS, mostly depends on the Q-network's training process.		
	The Q-network includes one input layer with $I$ neurons, one hidden layer with $H$ neurons, and two output layers with $V$ and $D$ neurons corresponding to the value stream and advantage stream, respectively. 
	Typically, DNN's training phase is conducted with matrix multiplication~\cite{goodfellow_deep_2016}.	
	Thus, the complexity of one training epoch of Q-network with a batch size $S_b$ is $\mathcal{O}\big(S_b(IH + HV + HD)\big)$.
	Supposed that Q-network's training phase consists of $T$ epochs, the complexity of i-ICS is $\mathcal{O}\big(TS_b(IH + HV + HD)\big)$. 	
	
	Generally, training a DNN requires a high computing resource, especially for complex neural architectures.
	Nevertheless, the Q-network in i-ICS only contains three layers in which only the hidden layer is fully connected.
	As such, our proposed approach can be effectively implemented on AVs that are usually equipped with sufficient computing resources. 
	In practice, there are some deep learning applications have been deployed in AVs such as Tesla Autopilot~\cite{autopilot} and ALVINN~\cite{pomerleau_alvinn_1989}.
	Thus, our proposed learning algorithm, i.e., i-ICS, can be effectively implemented on AVs. 
	
\section{Performance Evaluation}
	\label{sec:results}
	\subsection{Simulation Parameters}
	In this paper, we investigate our proposed solutions in a scenario that includes three types of objects: (i) ICS-AV, which is an autonomous vehicle equipped with 802.11ad based ICS, (ii) a receiving vehicle $AV_X$ that maintains communication with ICS-AV, and (iii) a target $AV_1$ moving at a distance of around the ICS-AV, as illustrated in Fig.~\ref{fig:system-model}.
	Specifically, the maximum related velocity that the ICS system can estimate is $50$~m/s.
	The CPI time $T_{CPI}$ is set to $10T_d$, meaning that the maximum number of frames in the CPI is 10.
	Note that one frame in the CPI can contain multiple packets.
	The ICS-AV has a data queue containing a maximum of $50$ packets.
	Recall that the data queue is used to store data packets when the amount of arrival data at the ICS transmitter excesses the ICS maximum transmission rate.
	All packets' sizes are assumed to be equal to $1500$ Bytes, which is the typical value of the maximum transfer unit (MTU) in WiFi networks and the Internet~\cite{nguyen_performance_2019}.	
	Other parameters of the ICS system are set based on IEEE 802.11ad standard, e.g., a carrier frequency of $60$ GHz and a sampling rate of $1.76$ GHz~\cite{802.11ad}.
	
	The 802.11ad standard requires that the packet error ratio (PER) to be less than $1\%$~\cite{802.11ad}.
	In practice, the PER depends on many factors such as wireless channel quality, modulation technique, and transmit power~\cite{khalili_a_new_2005}.
	Therefore, for demonstration purposes, we consider three levels with different values of PER: (i) level-1 with $10\%$, (ii) level-2 with $1\%$, and (iii) level-3 with $0.3\%$.
	Based on these quality levels, we assume that the wireless channel can fall into one of the three groups.
	The first one is the poor channel, in which the probability of channel quality at level-1, level-2, or level-3 at a time slot corresponds to a probability vector $\mathbf{p}_{c}^{p}= [0.6, 0.2, 0.2]$, i.e., at a time slot the probability of the channel quality at level 1 is 60\% and so on.
	Two other groups are normal channel and strong channel whose probability vectors are $\mathbf{p}_{c}^{n}= [0.2, 0.6, 0.2]$, and $\mathbf{p}_{c}^{g}= [0.2, 0.2, 0.6]$, respectively. 
	Note that these above settings are just for simulation purposes, our proposed learning algorithm (i.e., i-ICS) does not require to know these parameters in advance and can adapt with them through real-time interactions with the surrounding environment. 
	
	The parameters for the proposed learning algorithms are set as typical values, as in~\cite{Mnih2015Human,Hasselt2016_DoubleQ,Wang2016Dueling}. 
	Specifically, the value of $\epsilon$ in $\epsilon$-greedy policy is $1$ at the beginning of the learning process.
	Then, it is decreased at each time slot until $\epsilon=0.01$.
	The discount factor $\eta$ is $0.9$ for both the Q-learning based algorithm and the i-ICS.
	For the Q-learning based approach, the learning rate is $0.1$.
	For the proposed i-ICS, the Adam optimizer is used to train the Q-network with a learning rate of $10^{-4}$ and the target Q-network's parameters are updated at every $10^4$ time steps.
	
	Recall that the ICS-AV does not have any prior information about its surrounding environment's uncertainties and dynamics, e.g., the packet drop probabilities and packet arrival rate.
	Therefore, the proposed solutions are compared with two baseline policies: 1) a greedy policy where the ICS-AV selects an action to maximize the reward function without caring about the uncertainties and dynamics of the environment and 2) a deterministic policy in which the ICS-AV always sends $n_{dp}$ frames in the CPI time. 
	Here, we set $n_{dp}$ to a half of $N$ (i.e., $5$) to demonstrate the ICS's average performance when the number of frames in the CPI is fixed.
	Note that we do not consider conventional optimization-based methods (e.g.,~\cite{kumari_adaptive_2019, kumari_virtual_2018}) as baselines since they require complete information about the surrounding environment in advance to optimize the system parameters.
	
	\subsection{Simulation Results}
	To evaluate the ICS system performance, we first examine the convergence rates of our proposed approaches, i.e., Q-learning based algorithm and i-ICS. 
	We then study the influences of several key factors (e.g., the packet arrival rate, wireless channel quality, and weights in the immediate reward function) on the performance of the ICS system.
	\begin{figure}[t]
		\centering
		\includegraphics[width=0.65\linewidth]{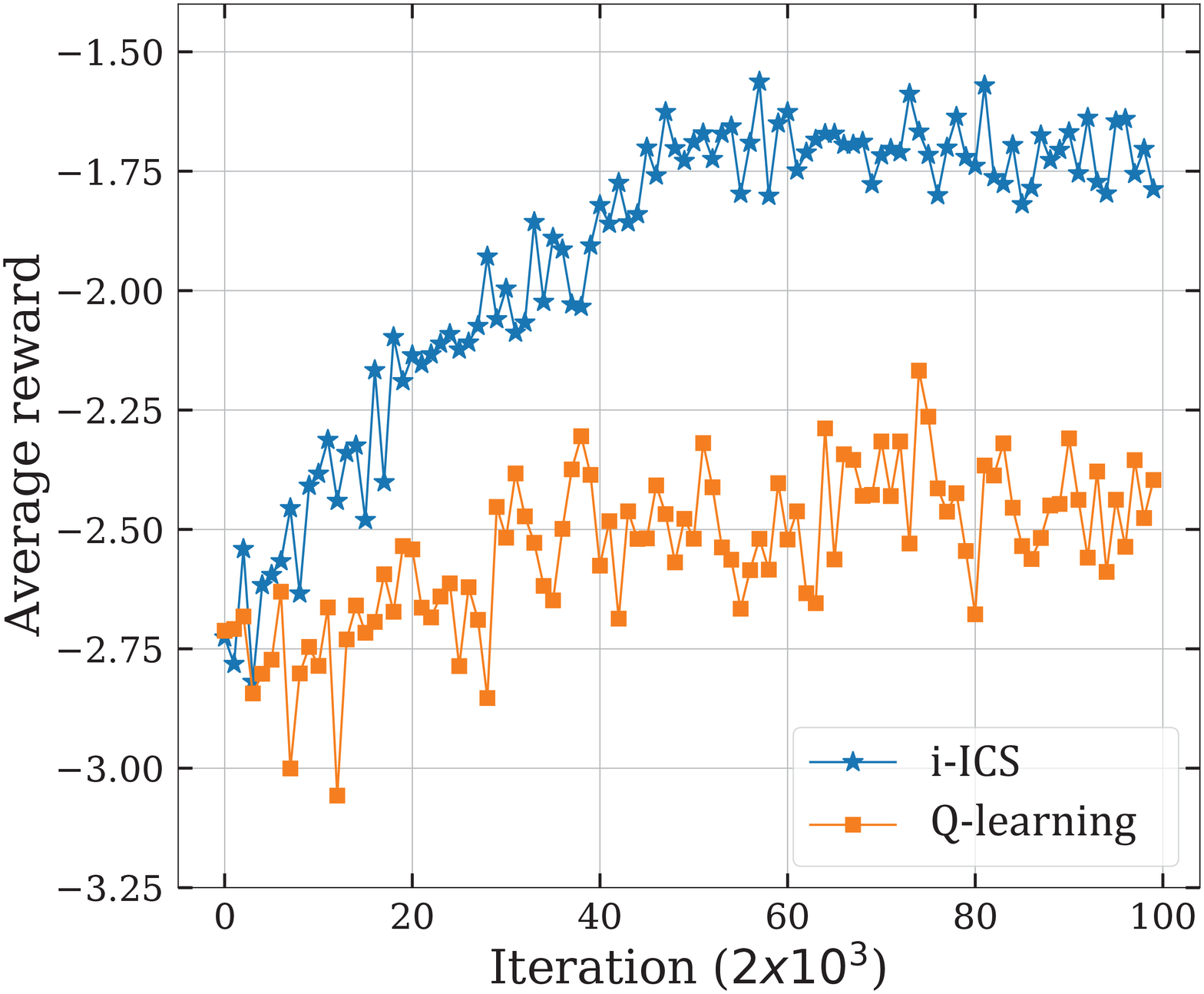}	
		\caption{Convergence rate of proposed algorithms.}
		\label{fig:convergene}
		\vspace{-10pt}
	\end{figure}
	\begin{figure*}[t]
		\centering
		$\begin{array}{cccc}
			\includegraphics[width=0.23\linewidth]{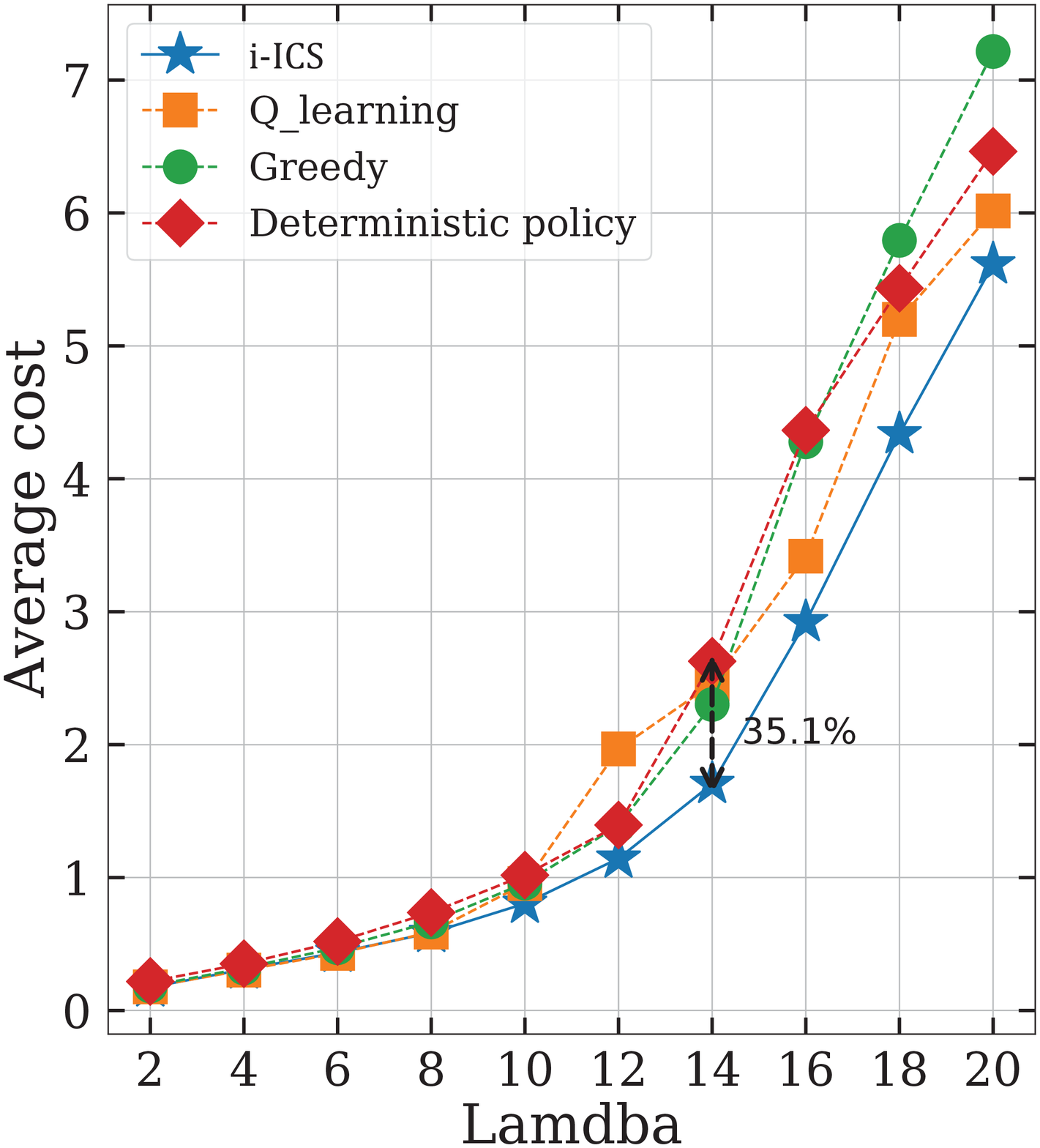}  
			&\includegraphics[width=0.23\linewidth]{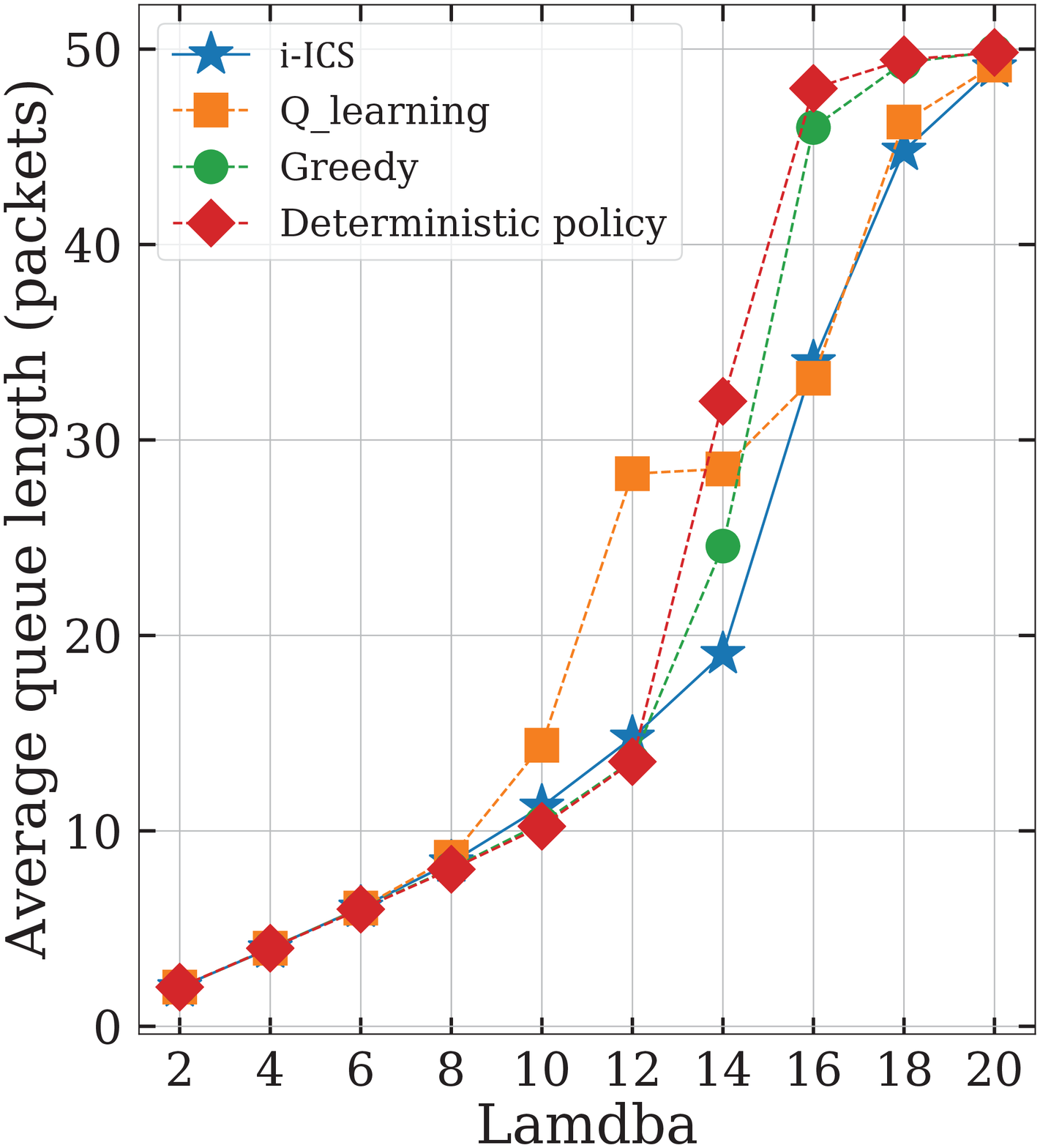}
			&\includegraphics[width=0.23\linewidth]{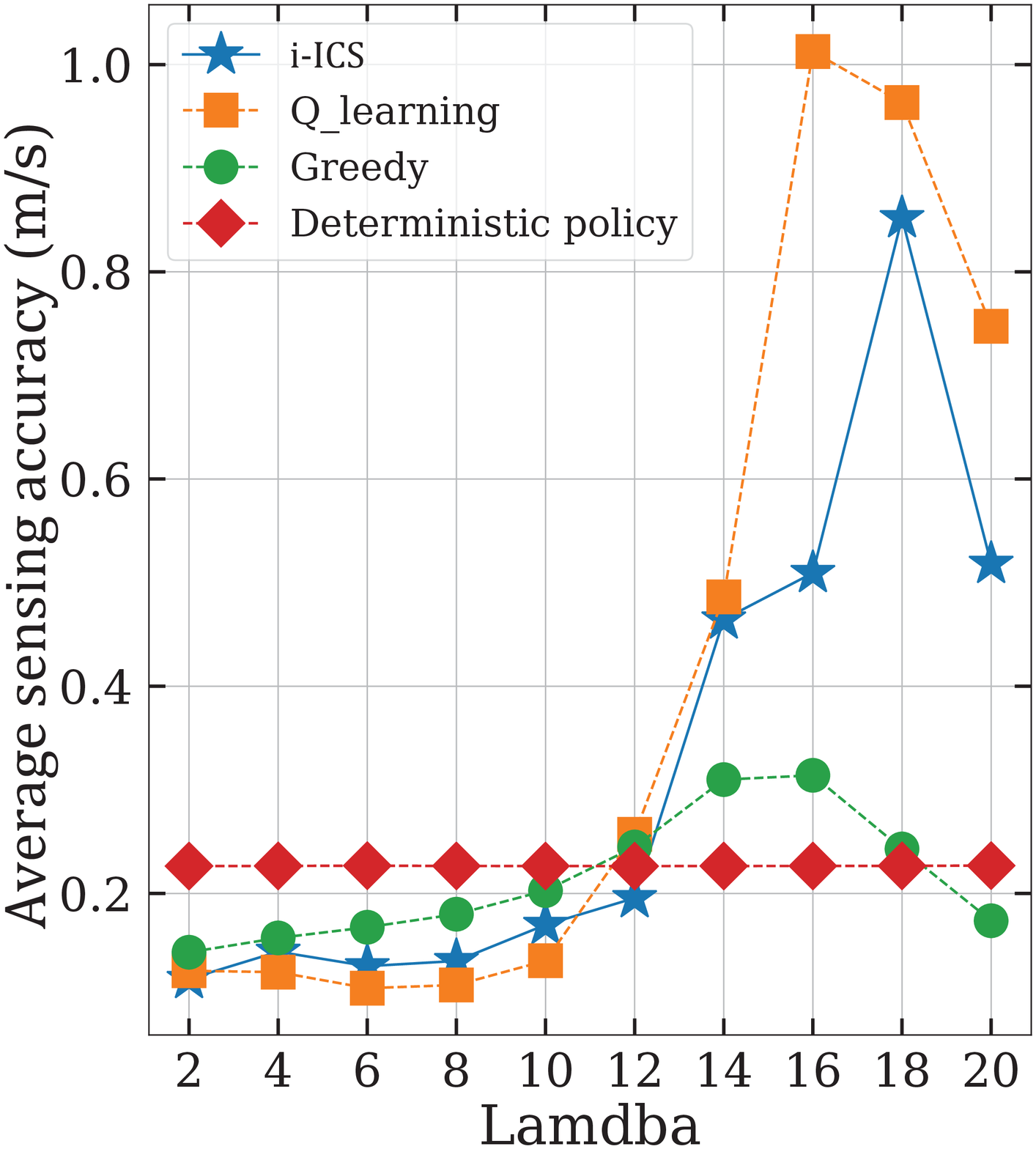}
			&\includegraphics[width=0.23\linewidth]{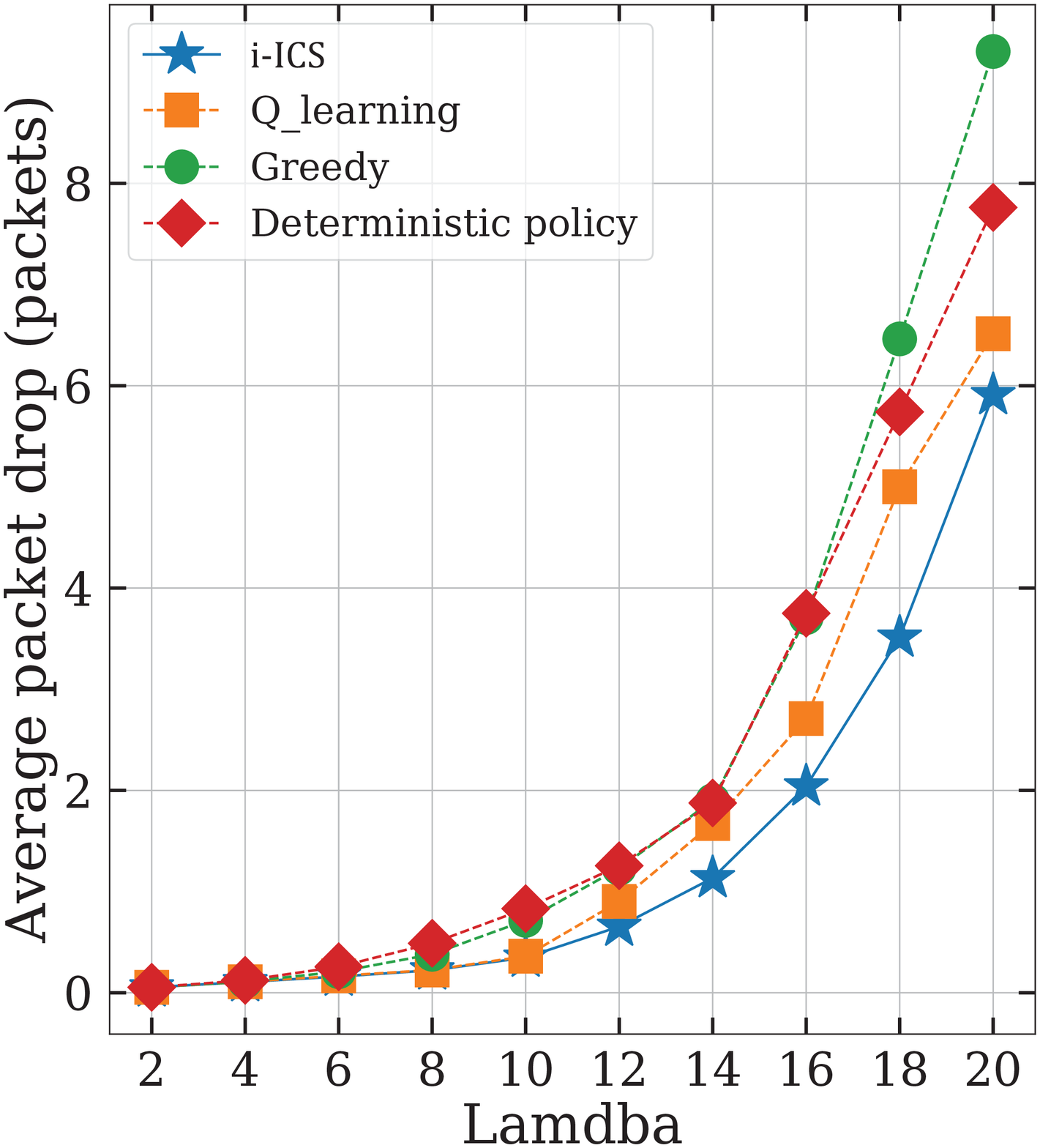}    
			\\
			{\fontsize{10}{1}\text{(a) Average cost}}
			&{\fontsize{10}{1}\text{(b) Average queue length}}
			&{\fontsize{10}{1}\text{(c) Average sensing accuracy}}
			&{\fontsize{10}{1}\text{(d) Average packet drop}}
		\end{array}$
		\caption{Varying data arrival rate with normal channel quality, $w_1=0.05,w_2=0.4$, and $w_3=0.5$.}
		\label{fig:normal_w1}
	\end{figure*}	
	\begin{figure*}[t]
		\centering
		$\begin{array}{cccc}
			\includegraphics[width=0.23\linewidth]{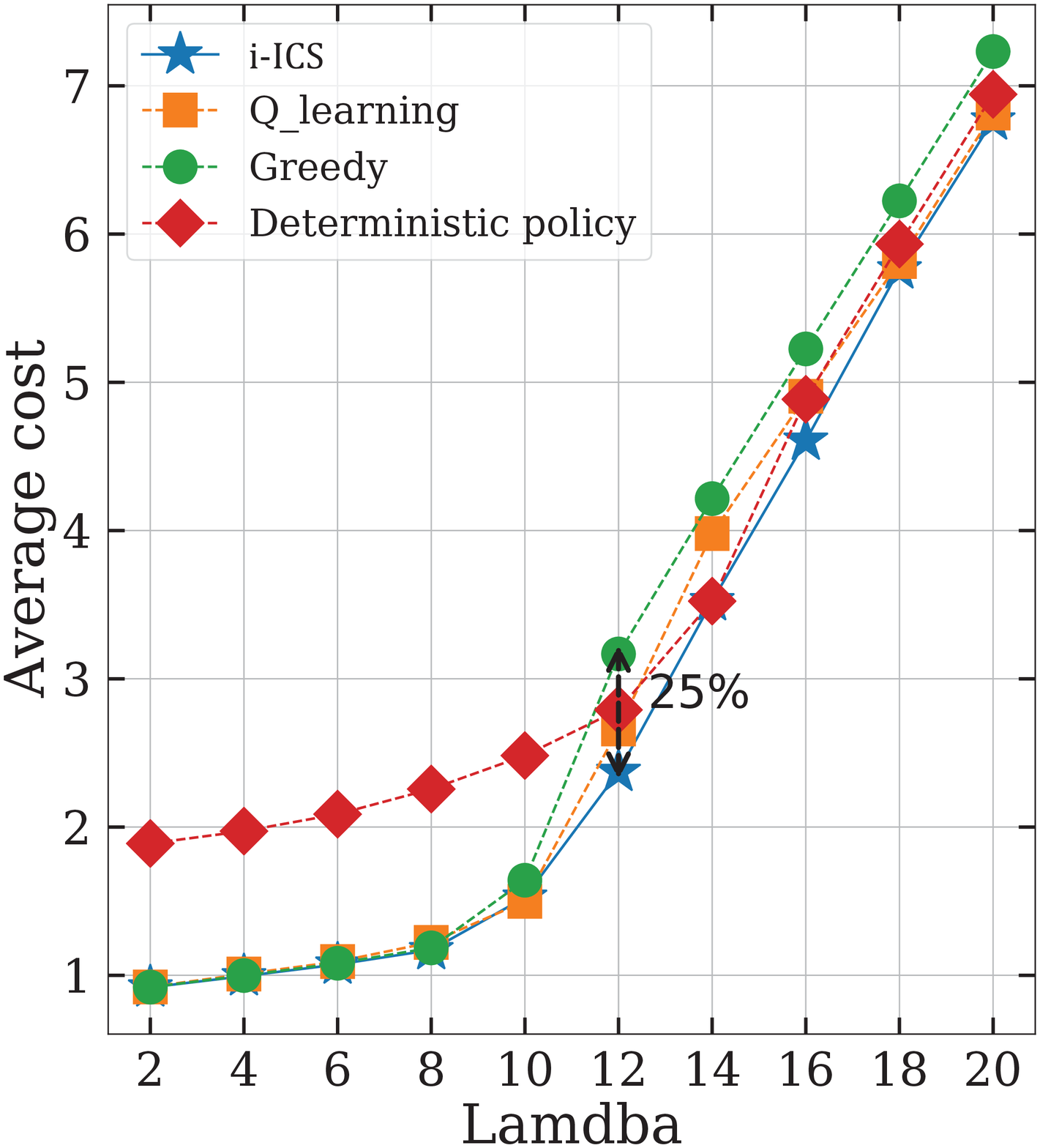}  
			&\includegraphics[width=0.23\linewidth]{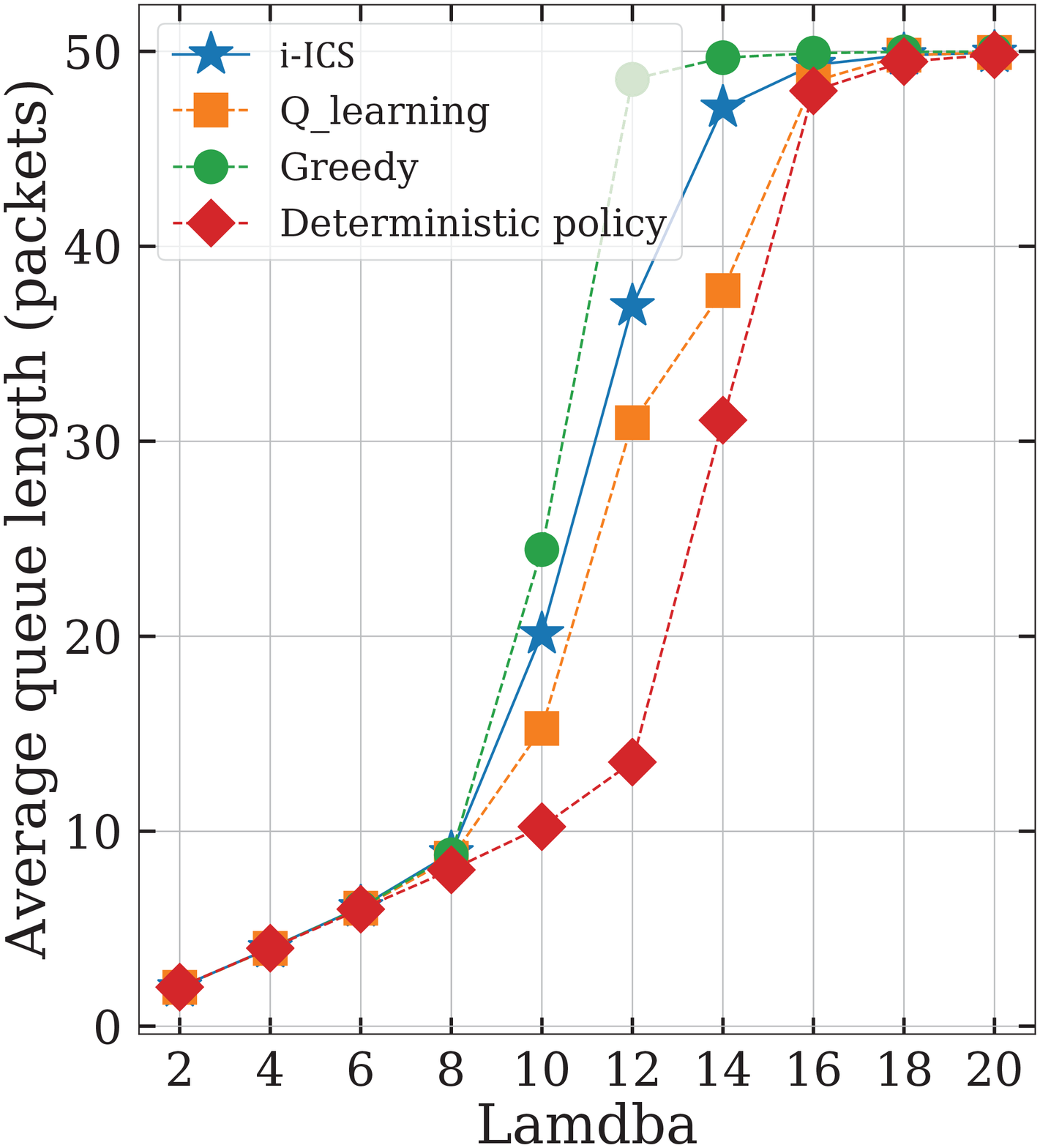}
			&\includegraphics[width=0.23\linewidth]{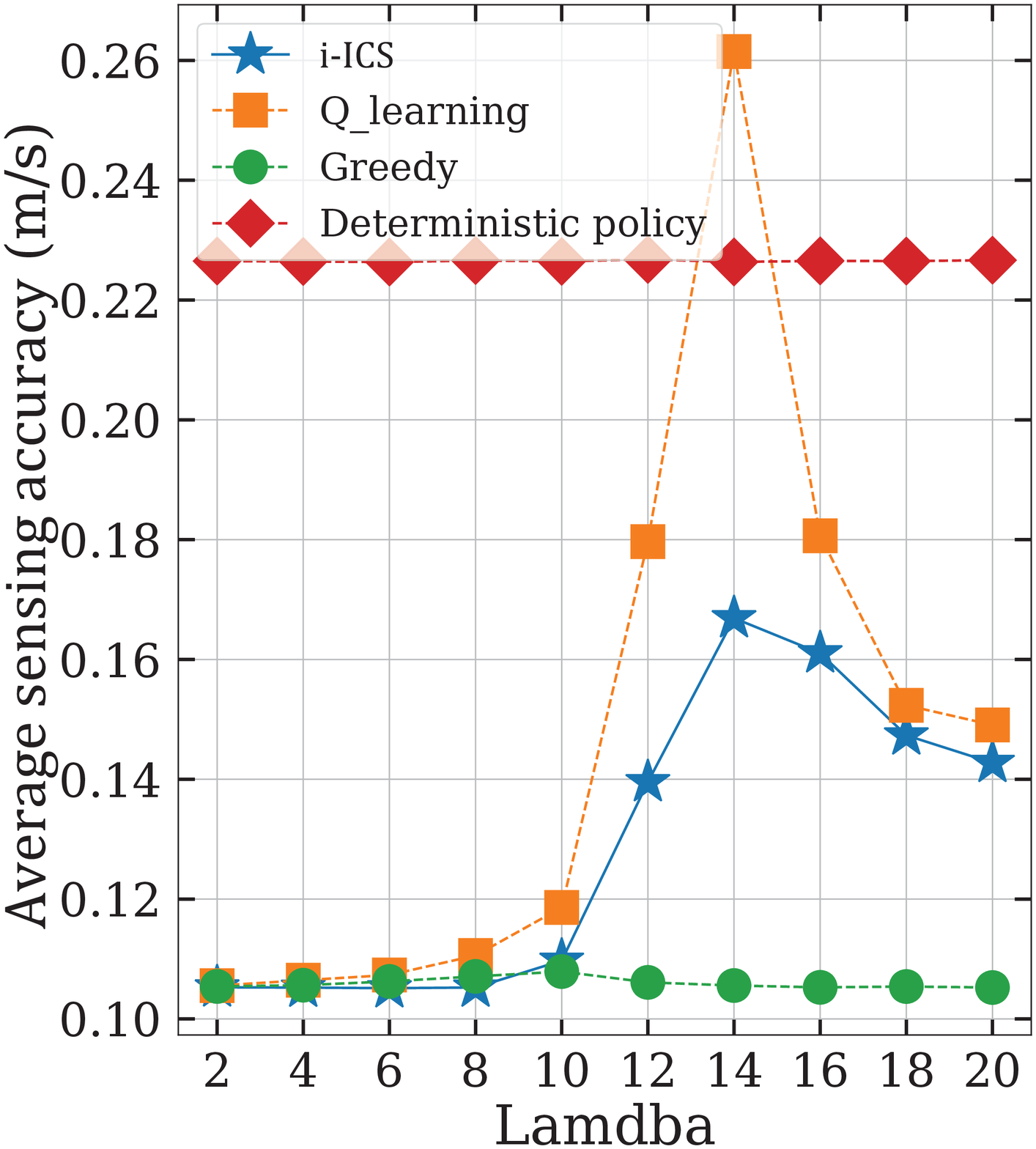}
			&\includegraphics[width=0.23\linewidth]{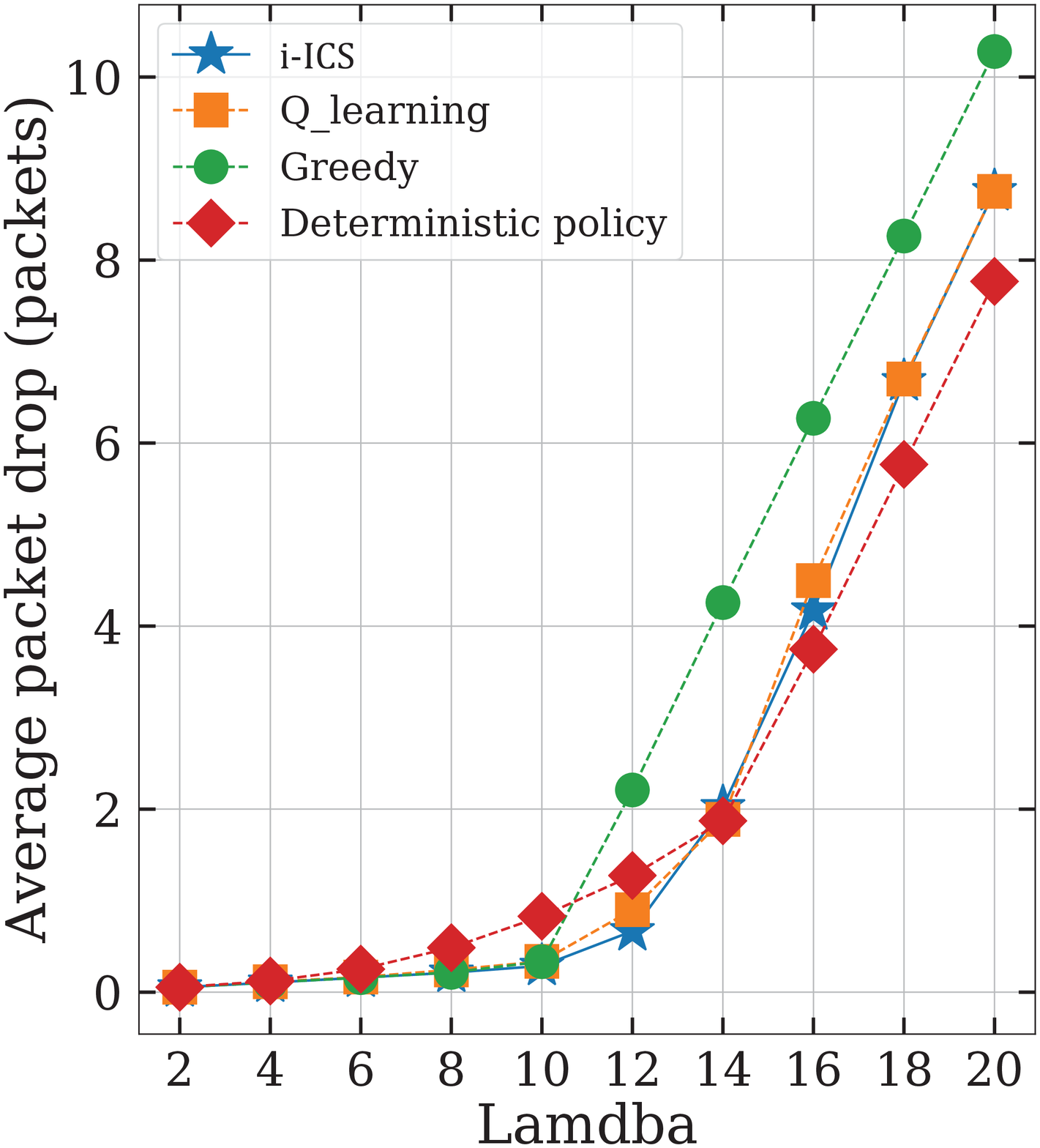}    
			\\
			{\fontsize{10}{1}\text{(a) Average cost}}
			&{\fontsize{10}{1}\text{(b) Average queue length}}
			&{\fontsize{10}{1}\text{(c) Average sensing accuracy}}
			&{\fontsize{10}{1}\text{(d) Average packet drop}}
		\end{array}$
		\caption{Varying data arrival rate with normal channel quality, $w_1=0.025,w_2=0.8$, and $w_3=0.5$.}
		\label{fig:normal_w2}
	\end{figure*}
	
	\subsubsection{Convergence Rate}
	Fig.~\ref{fig:convergene} illustrates the convergence rates of our proposed algorithms for the ICS system, i.e., the Q-learning and the i-ICS.
	Here, we compare their performance in the normal channel, and the mean number of arrived packets $\lambda$ is set to $14$.
	It can be observed that the i-ICS achieves a superior result in terms of average reward compared with that of the Q-learning.
	Specifically, at the beginning of the learning process, the Q-learning and i-ICS obtain similar results.
	However, after $2$$\times$$10^4$ iterations, the i-ICS's average reward is 20\% greater than that of the Q-learning.
	Then, the i-ICS eventually converges to the optimal policy after $4.5\!\times\!10^4$ while Q-learning still struggles with a mediocre policy.
	Under the optimal policy obtained by i-ICS, the ICS-AV's average reward is stable at around $-1.75$, approximately $40\%$ higher than that of the policy learned by the Q-learning.
	\begin{figure*}[t]
		\centering
		$\begin{array}{cccc}
			\includegraphics[width=0.23\linewidth]{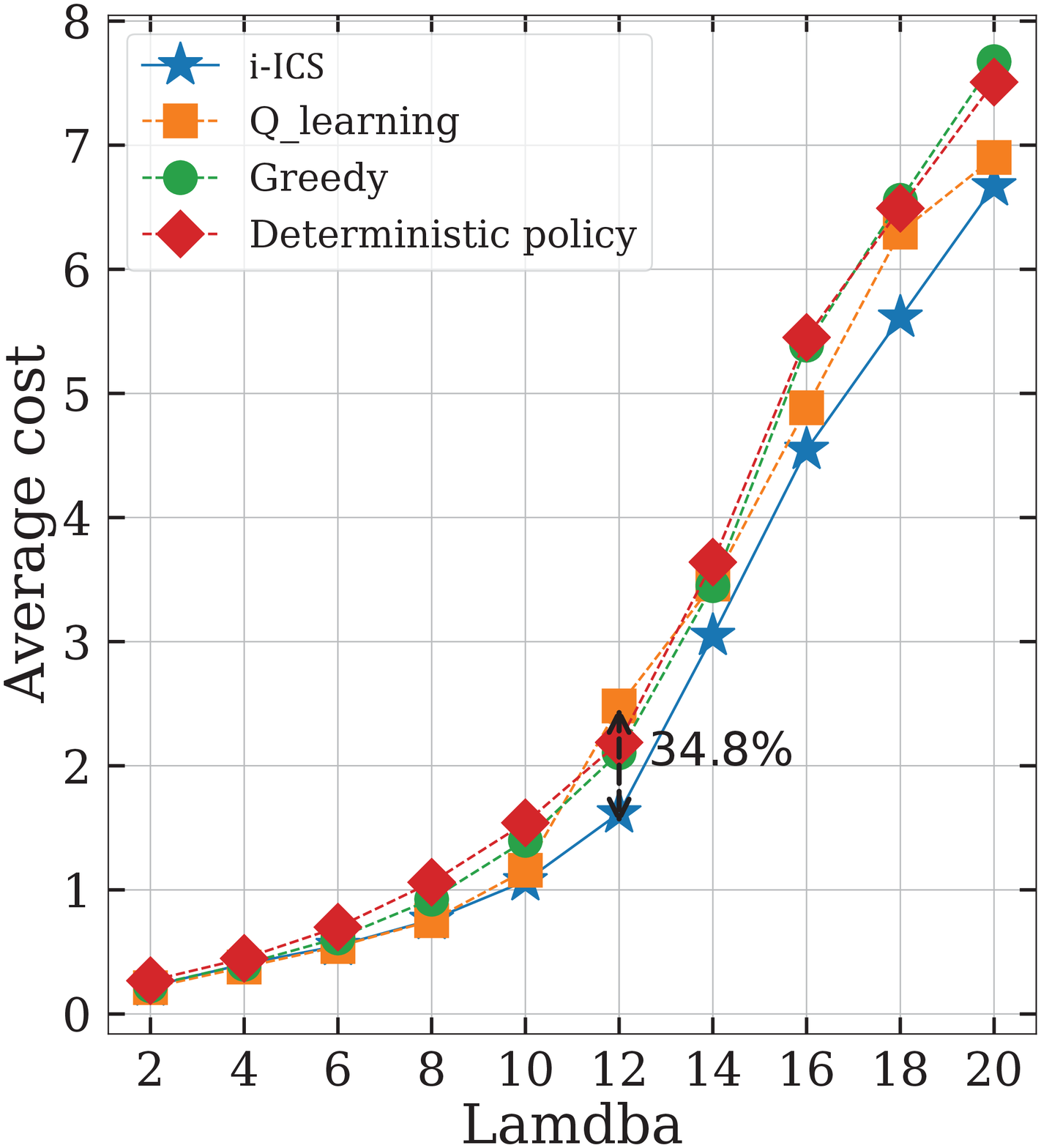}  
			&\includegraphics[width=0.23\linewidth]{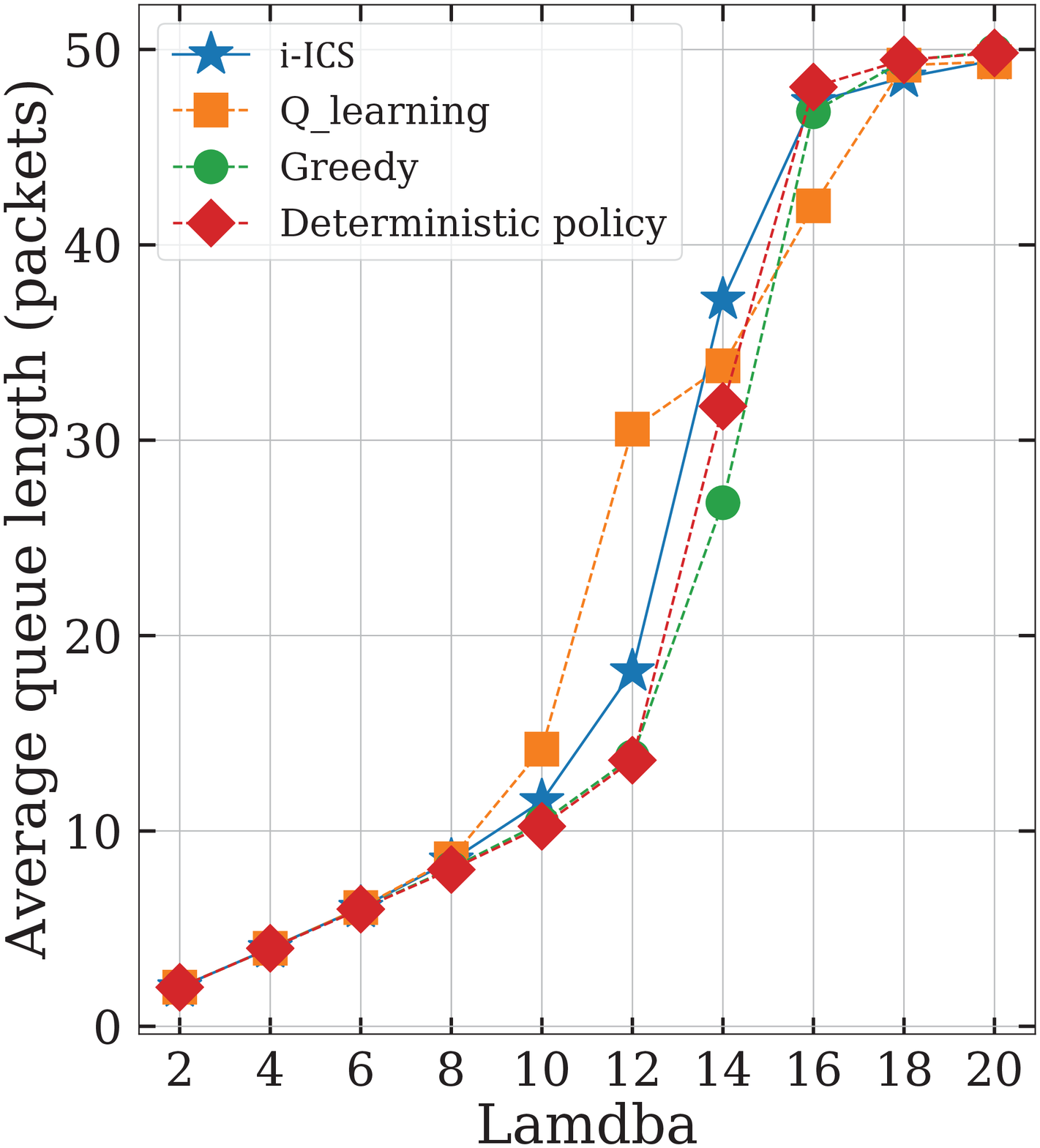}
			&\includegraphics[width=0.23\linewidth]{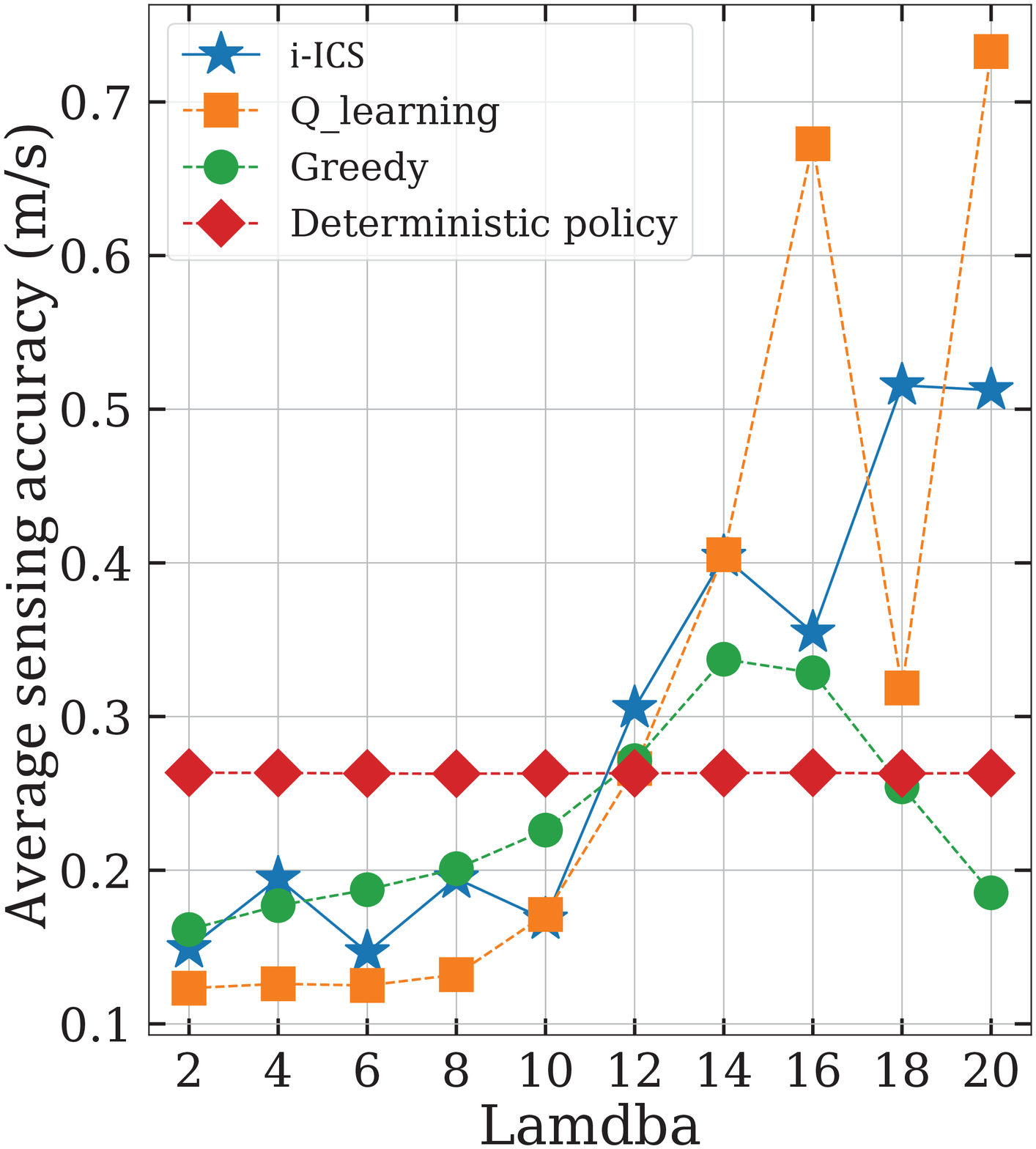}
			&\includegraphics[width=0.23\linewidth]{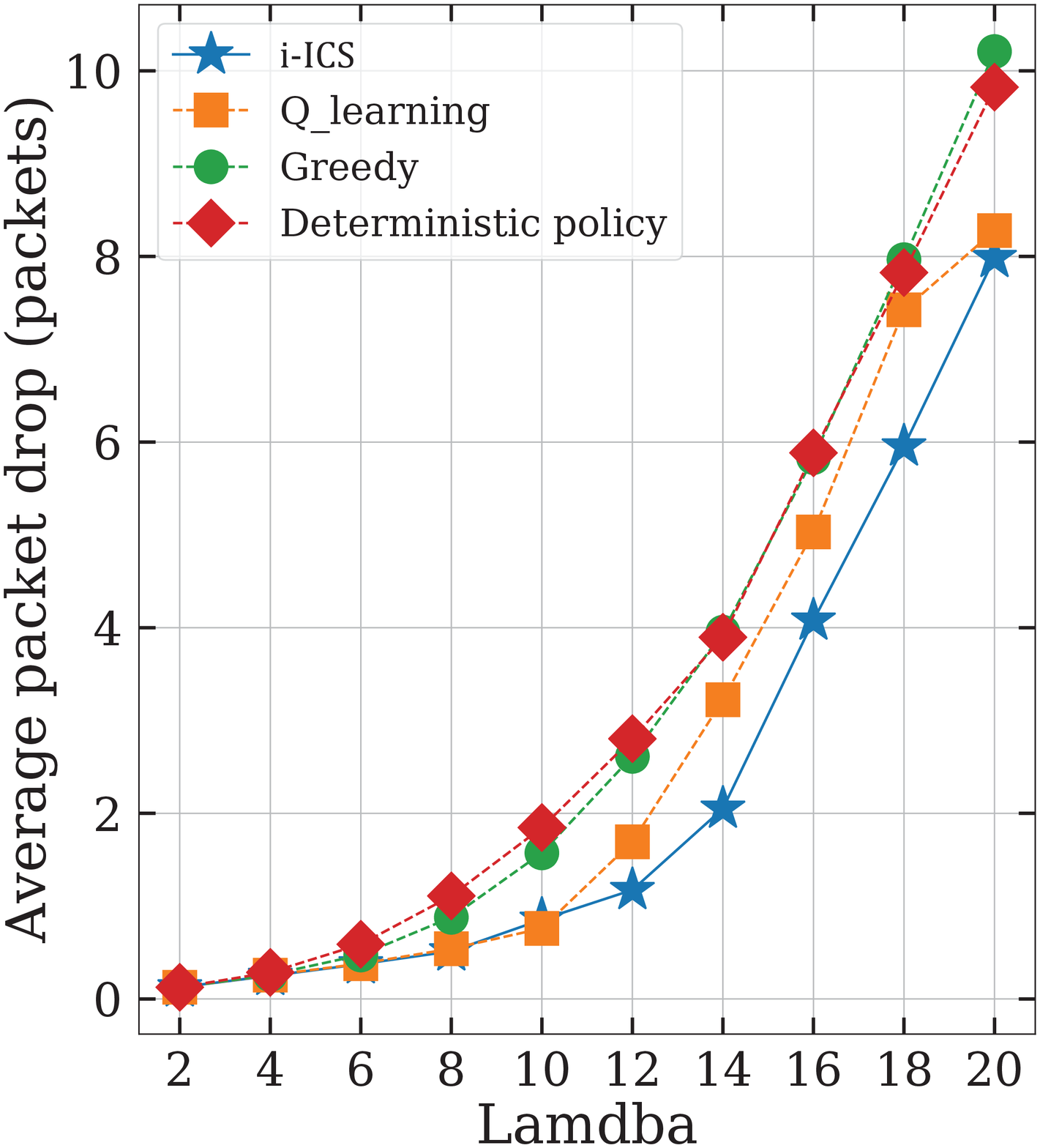}    
			\\
			{\fontsize{10}{1}\text{(a) Average cost}}
			&{\fontsize{10}{1}\text{(b) Average queue length}}
			&{\fontsize{10}{1}\text{(c) Average sensing accuracy}}
			&{\fontsize{10}{1}\text{(d) Average packet drop}}
		\end{array}$
		\caption{Varying data arrival rate with poor channel quality, $w_1=0.05,w_2=0.4$, and $w_3=0.5$.}
		\label{fig:bad_w1}
		\vspace{-10pt}
	\end{figure*}
	\begin{figure*}
		\centering
		$\begin{array}{cccc}
			\includegraphics[width=0.23\linewidth]{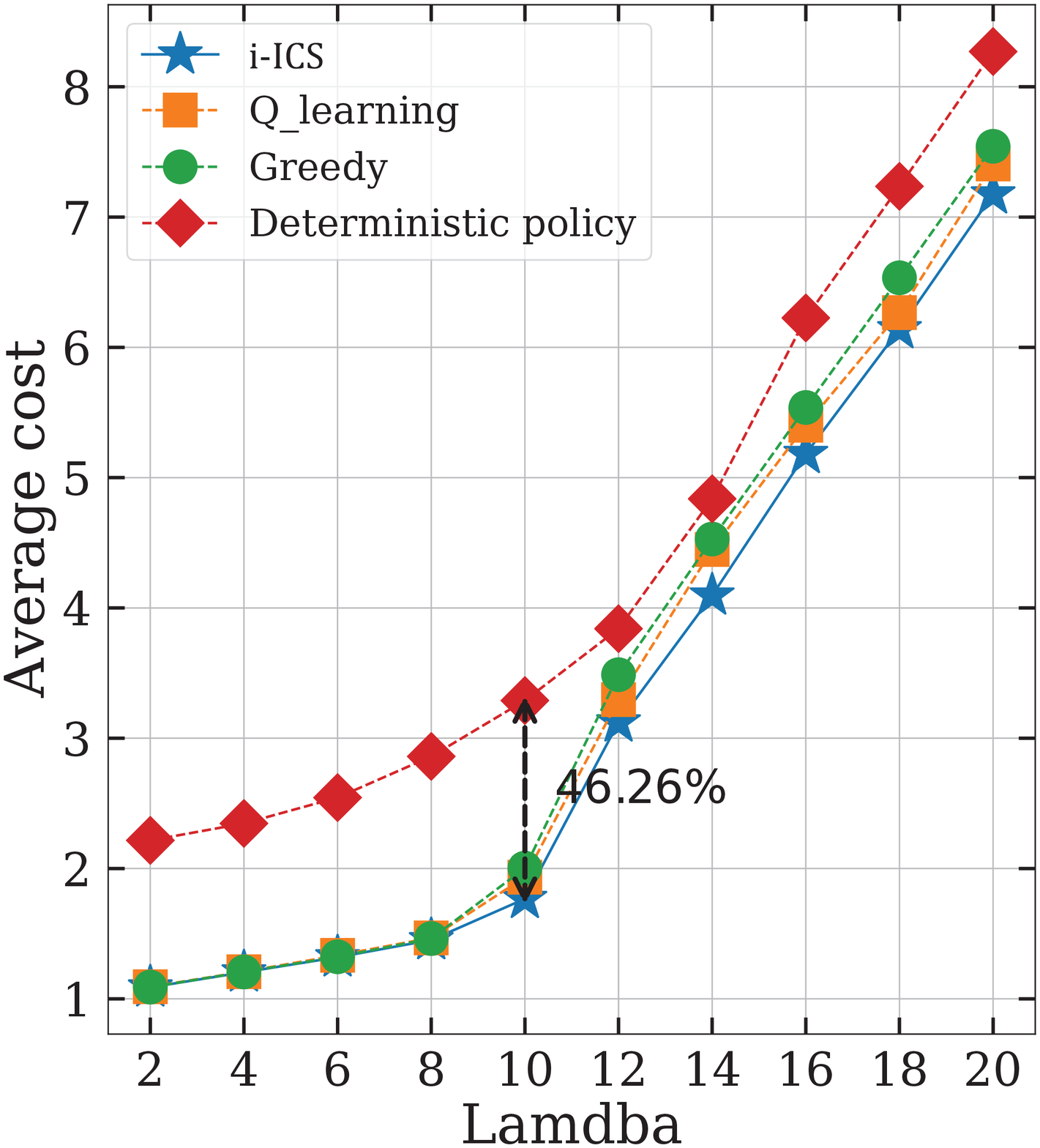}  
			&\includegraphics[width=0.23\linewidth]{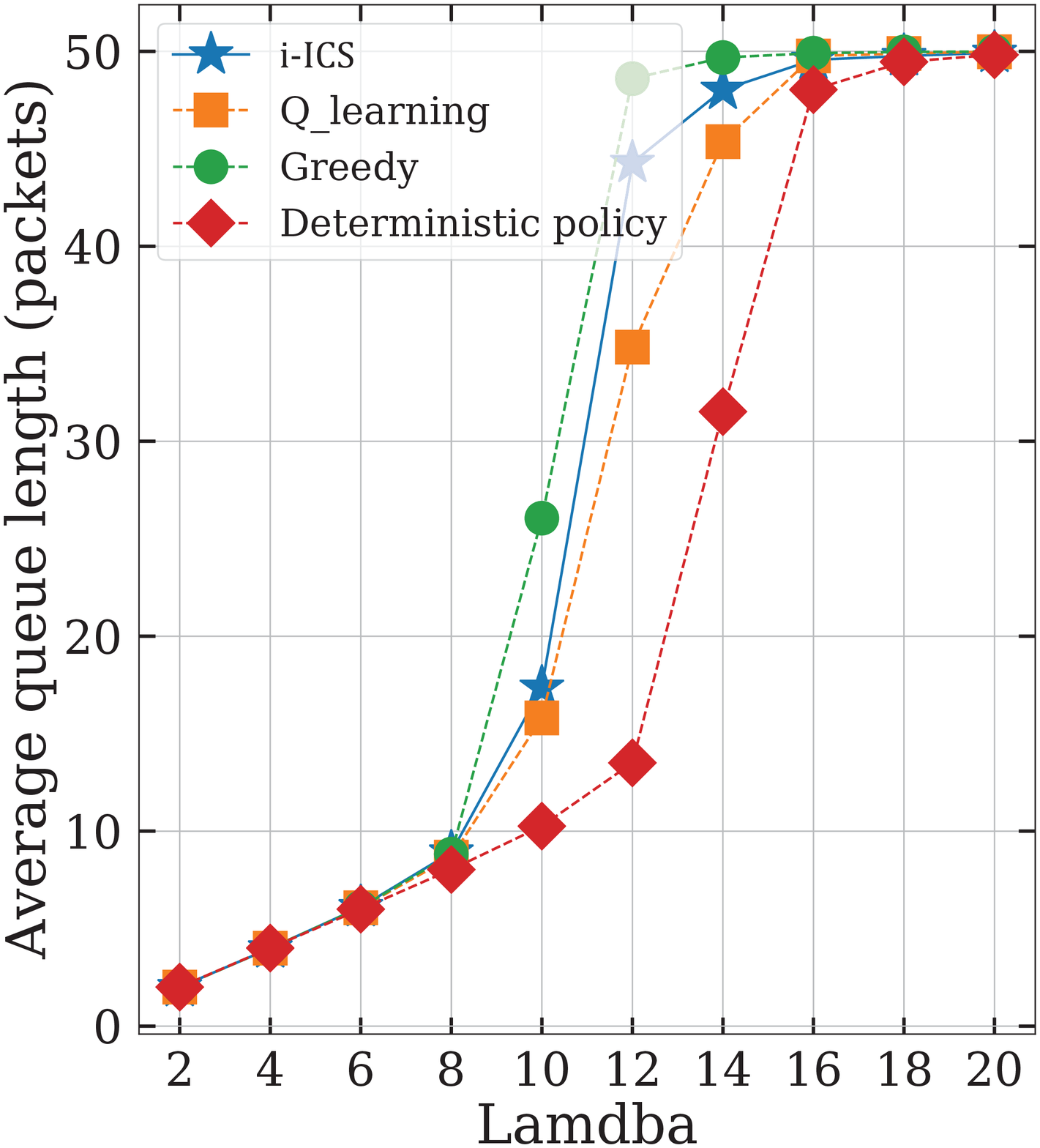}
			&\includegraphics[width=0.23\linewidth]{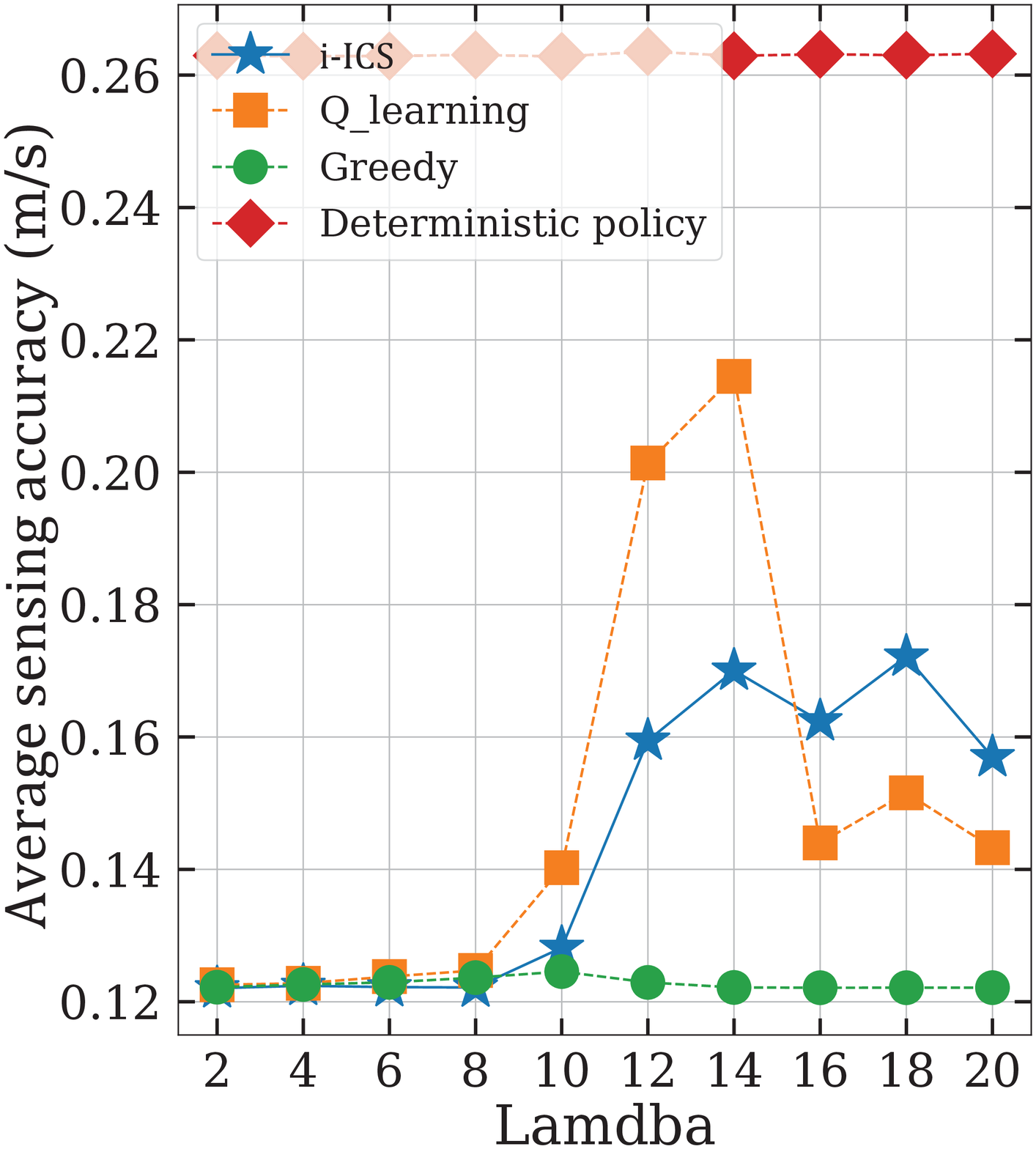}
			&\includegraphics[width=0.23\linewidth]{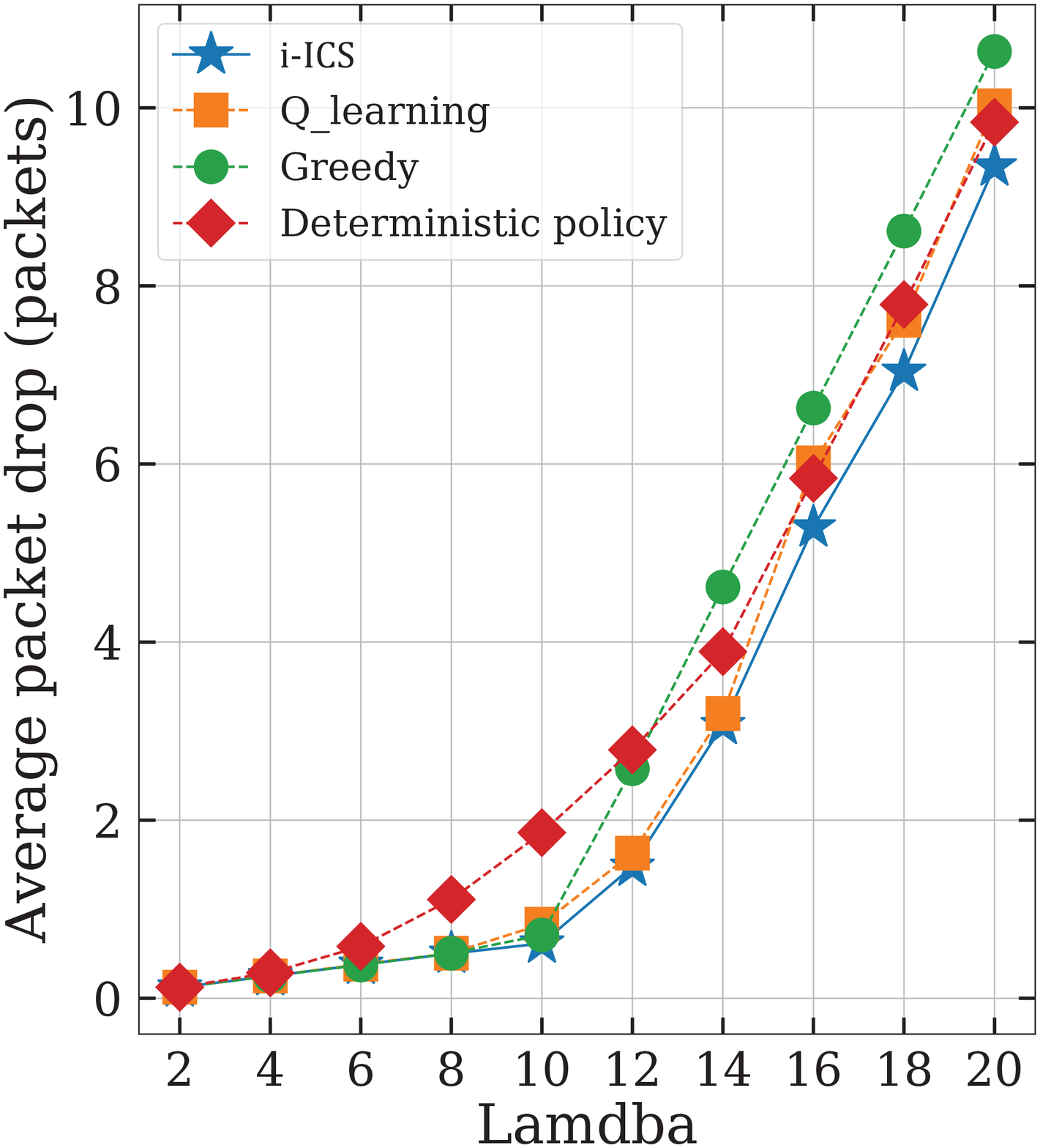}    
			\\
			{\fontsize{10}{1}\text{(a) Average cost}}
			&{\fontsize{10}{1}\text{(b) Average queue length}}
			&{\fontsize{10}{1}\text{(c) Average sensing accuracy}}
			&{\fontsize{10}{1}\text{(d) Average packet drop}}
		\end{array}$
		\caption{Varying data arrival rate with poor channel quality, $w_1=0.025,w_2=0.8$, and $w_3=0.5$.}
		\label{fig:bad_w2}
	\end{figure*}	
	\begin{figure*}[t]
		\centering
		$\begin{array}{cccc}
			\includegraphics[width=0.23\linewidth]{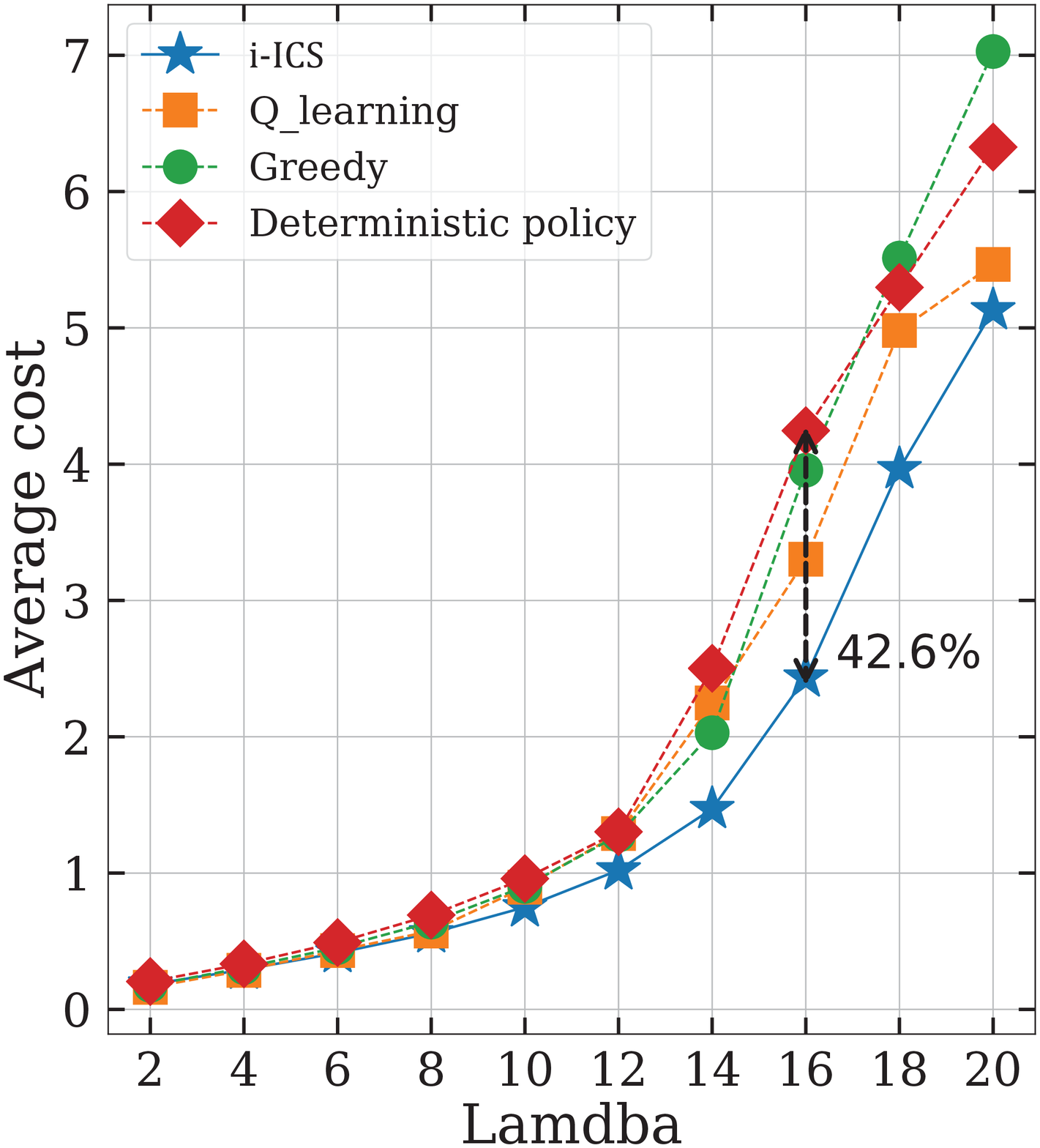}  
			&\includegraphics[width=0.23\linewidth]{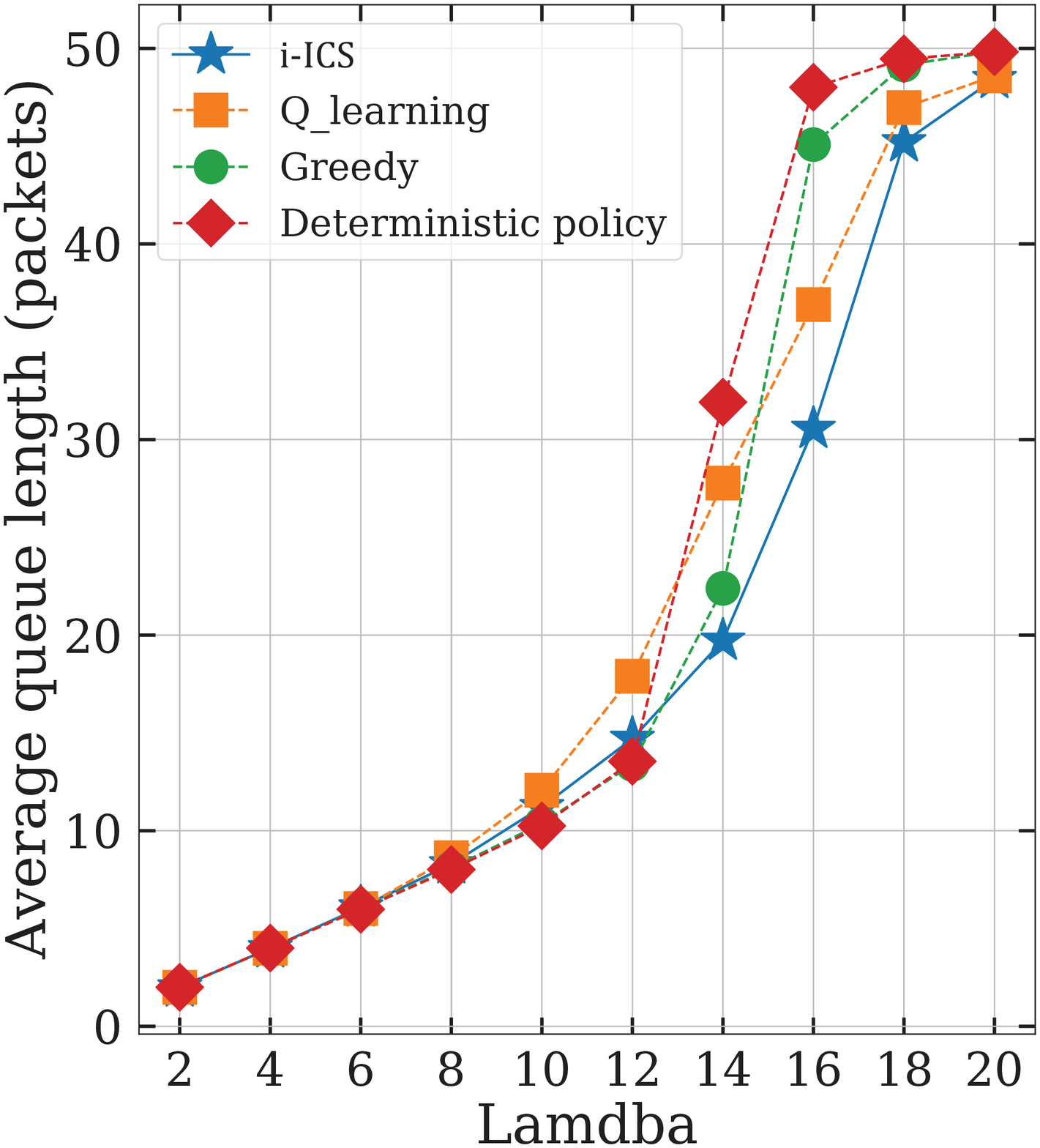}
			&\includegraphics[width=0.23\linewidth]{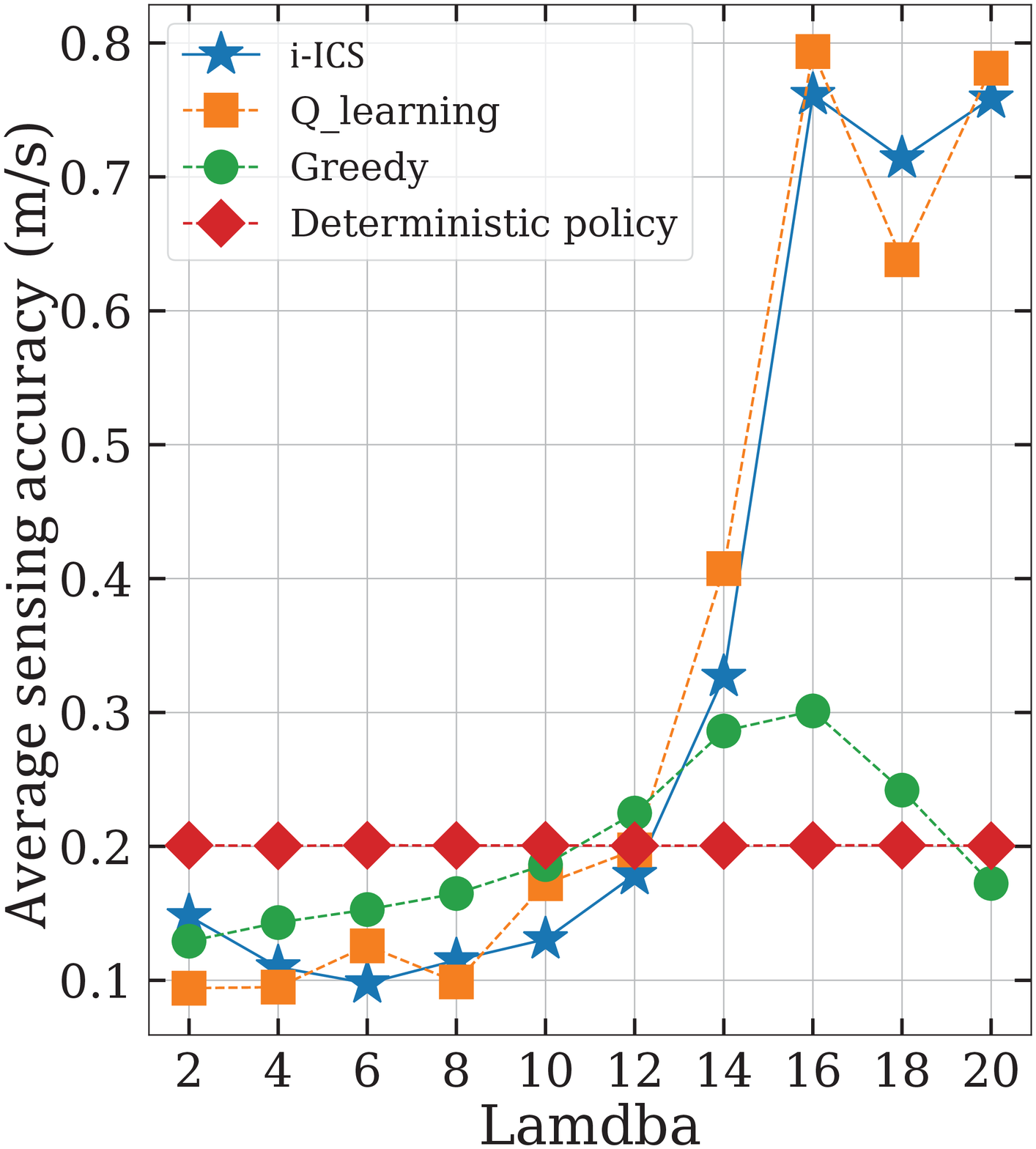}
			&\includegraphics[width=0.23\linewidth]{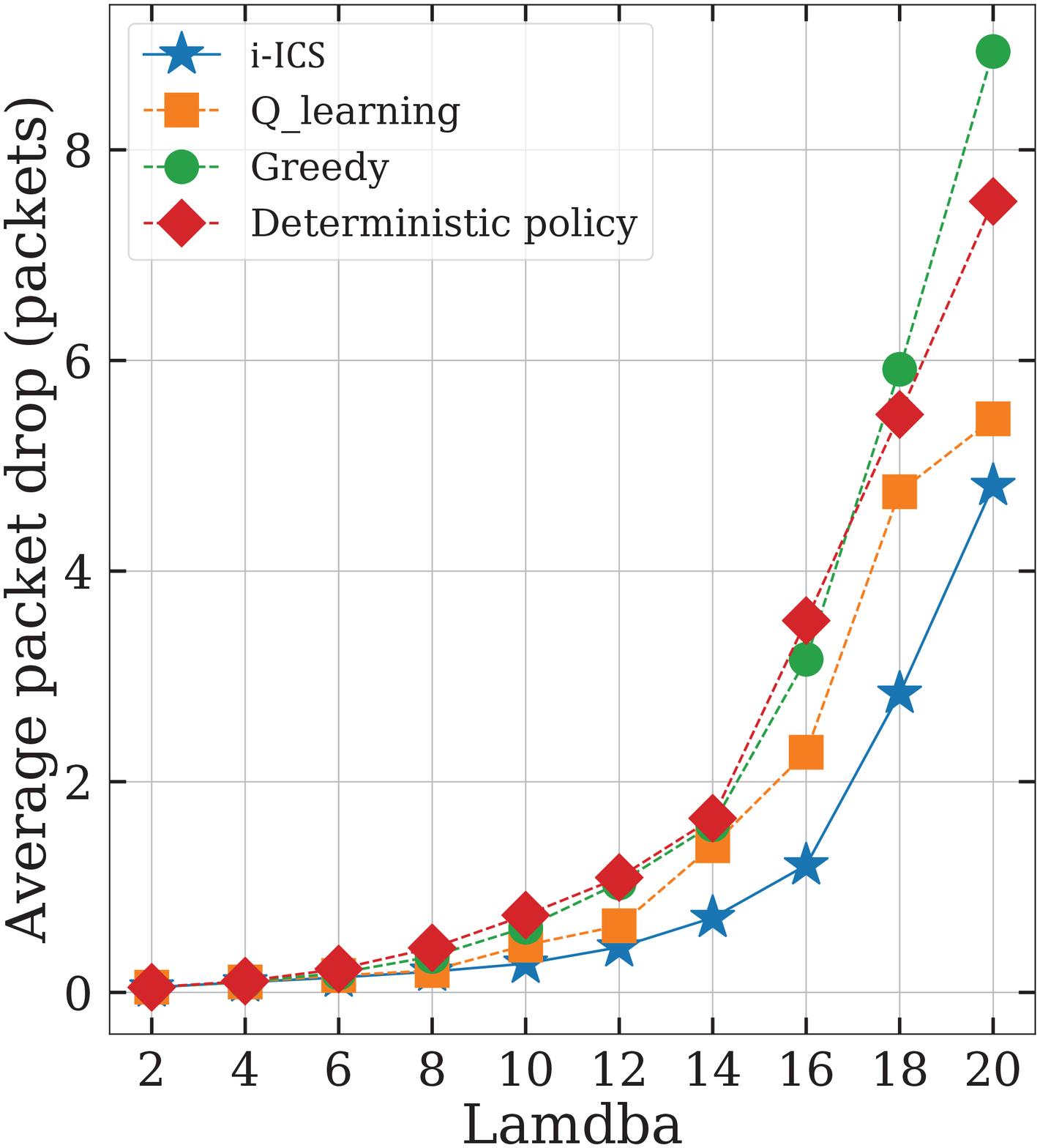}    
			\\
			{\fontsize{10}{1}\text{(a) Average cost}}
			&{\fontsize{10}{1}\text{(b) Average queue length}}
			&{\fontsize{10}{1}\text{(c) Average sensing accuracy}}
			&{\fontsize{10}{1}\text{(d) Average packet drop}}
		\end{array}$
		\caption{Varying data arrival rate with strong channel quality, $w_1=0.05,w_2=0.4$, and $w_3=0.5$.}
		\label{fig:good_w1}
	\end{figure*}
	
	\begin{figure*}[t]
		\centering
		$\begin{array}{cccc}
			\includegraphics[width=0.23\linewidth]{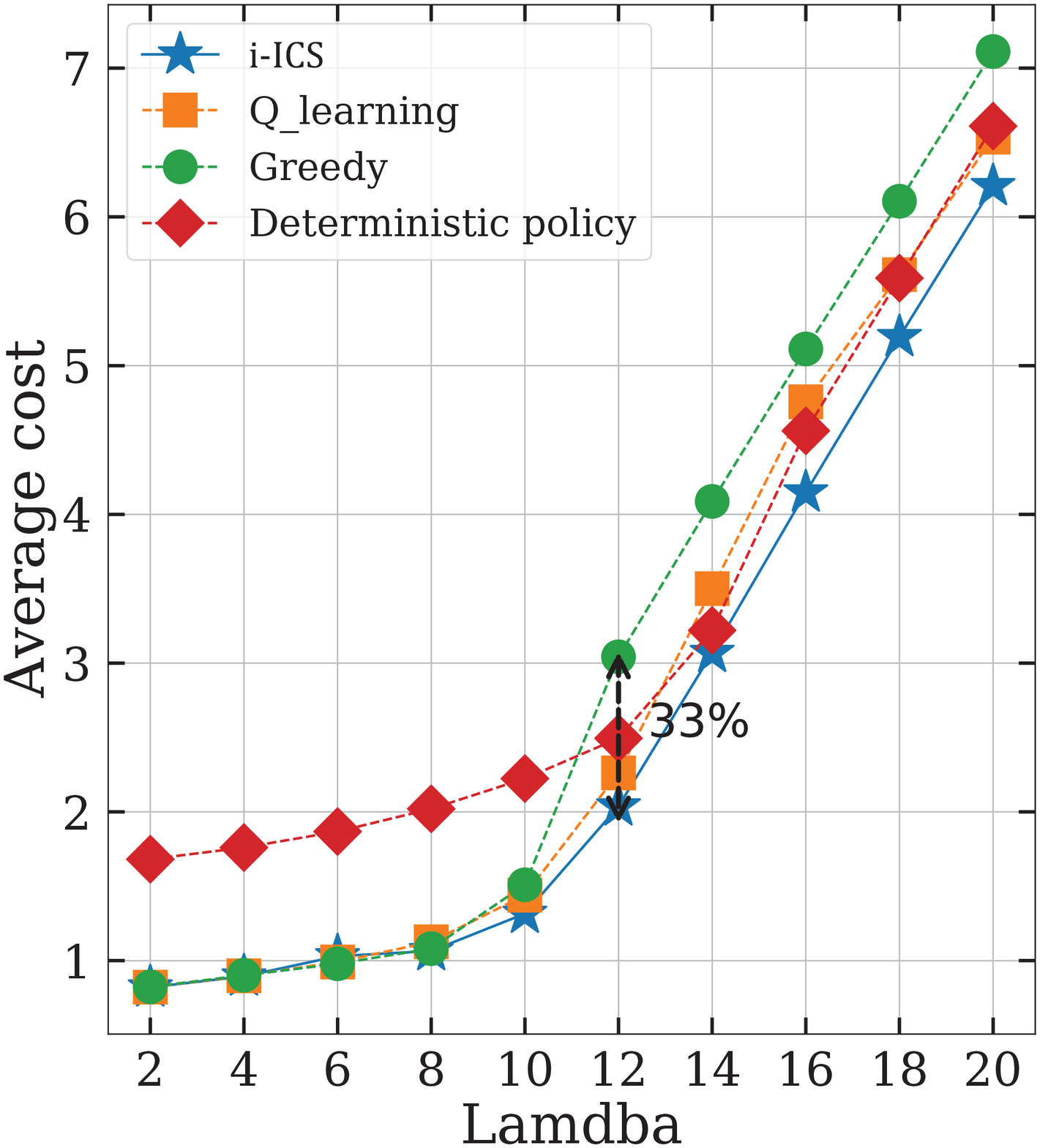}  
			&\includegraphics[width=0.23\linewidth]{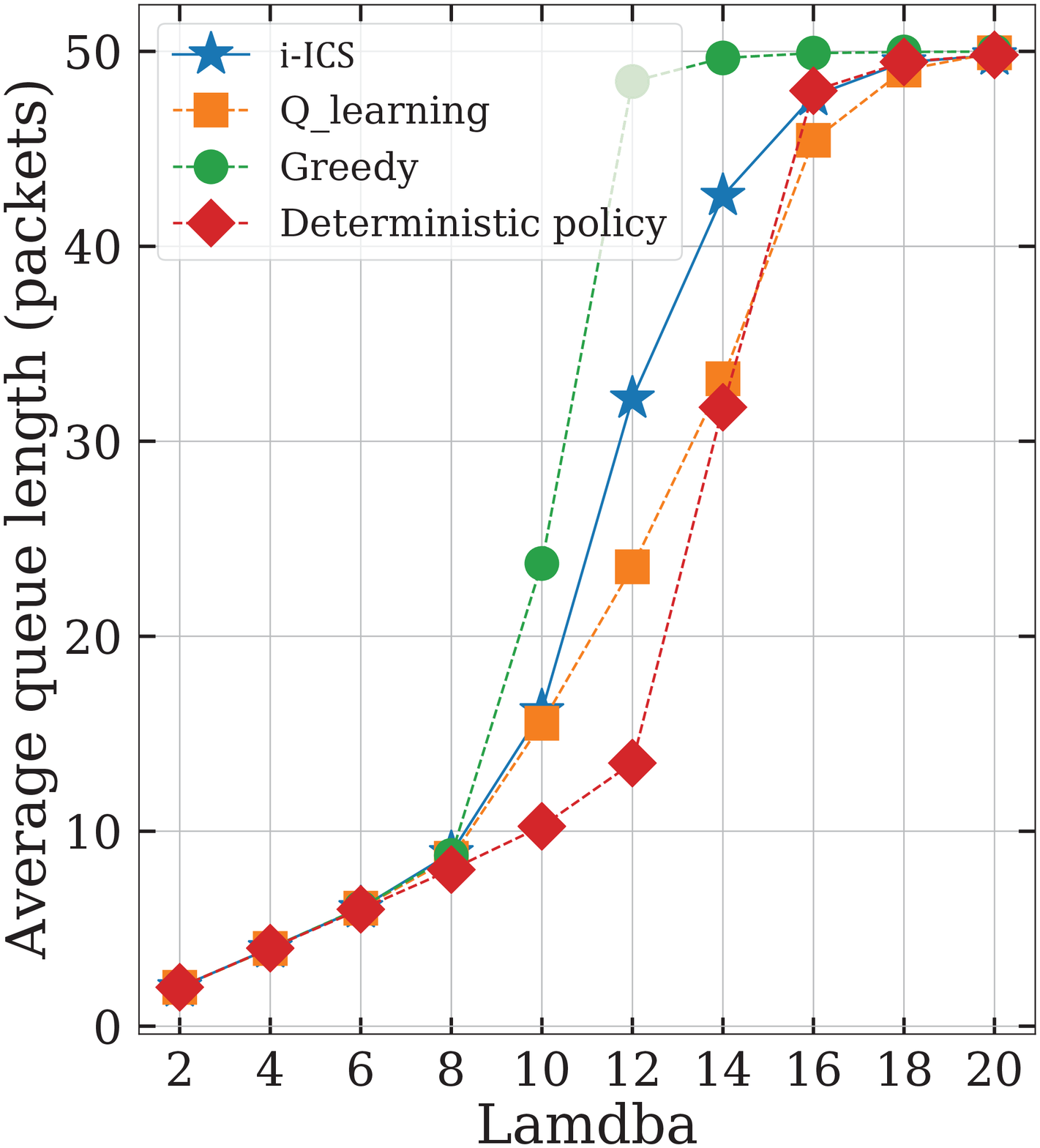}
			&\includegraphics[width=0.23\linewidth]{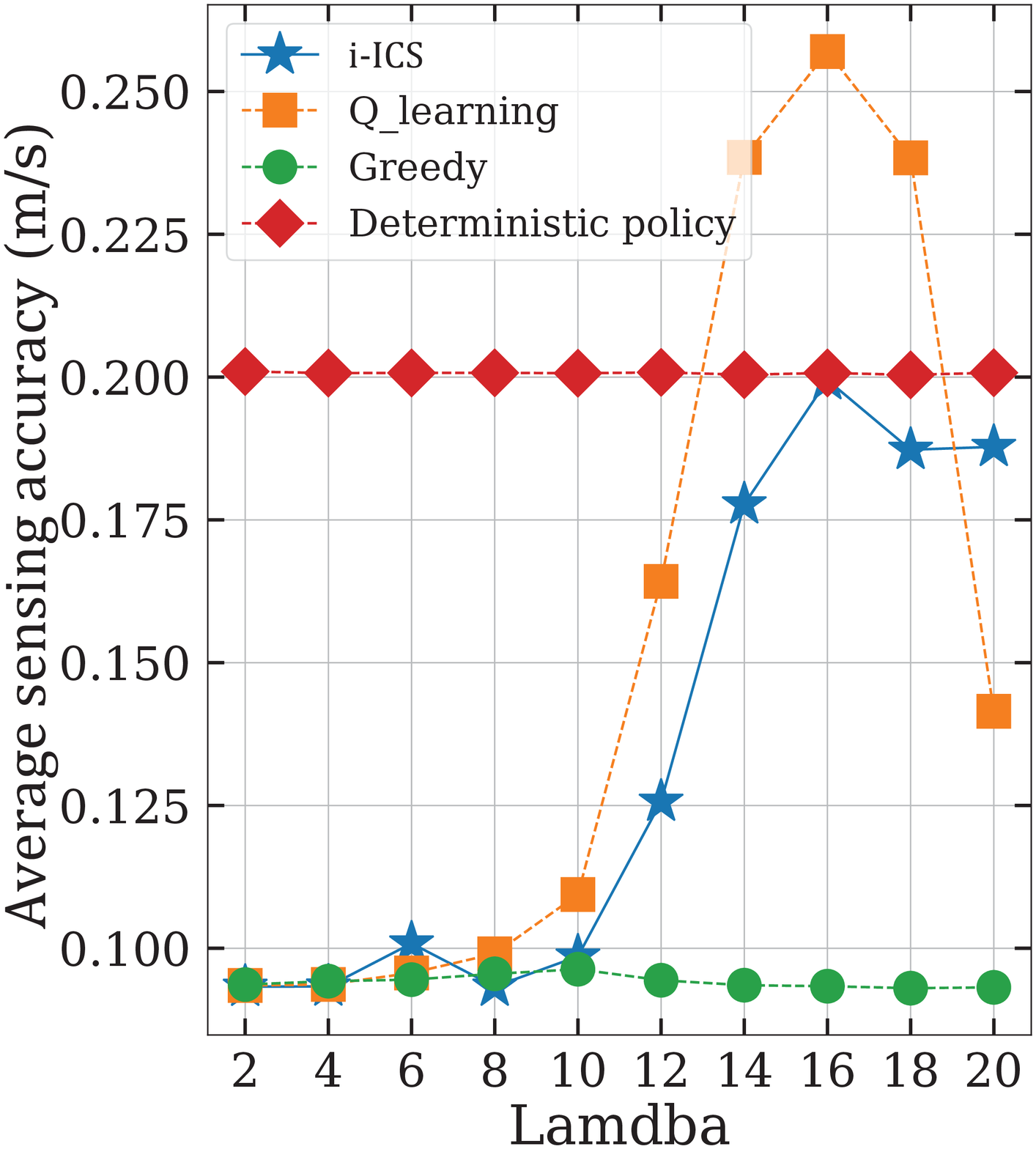}
			&\includegraphics[width=0.23\linewidth]{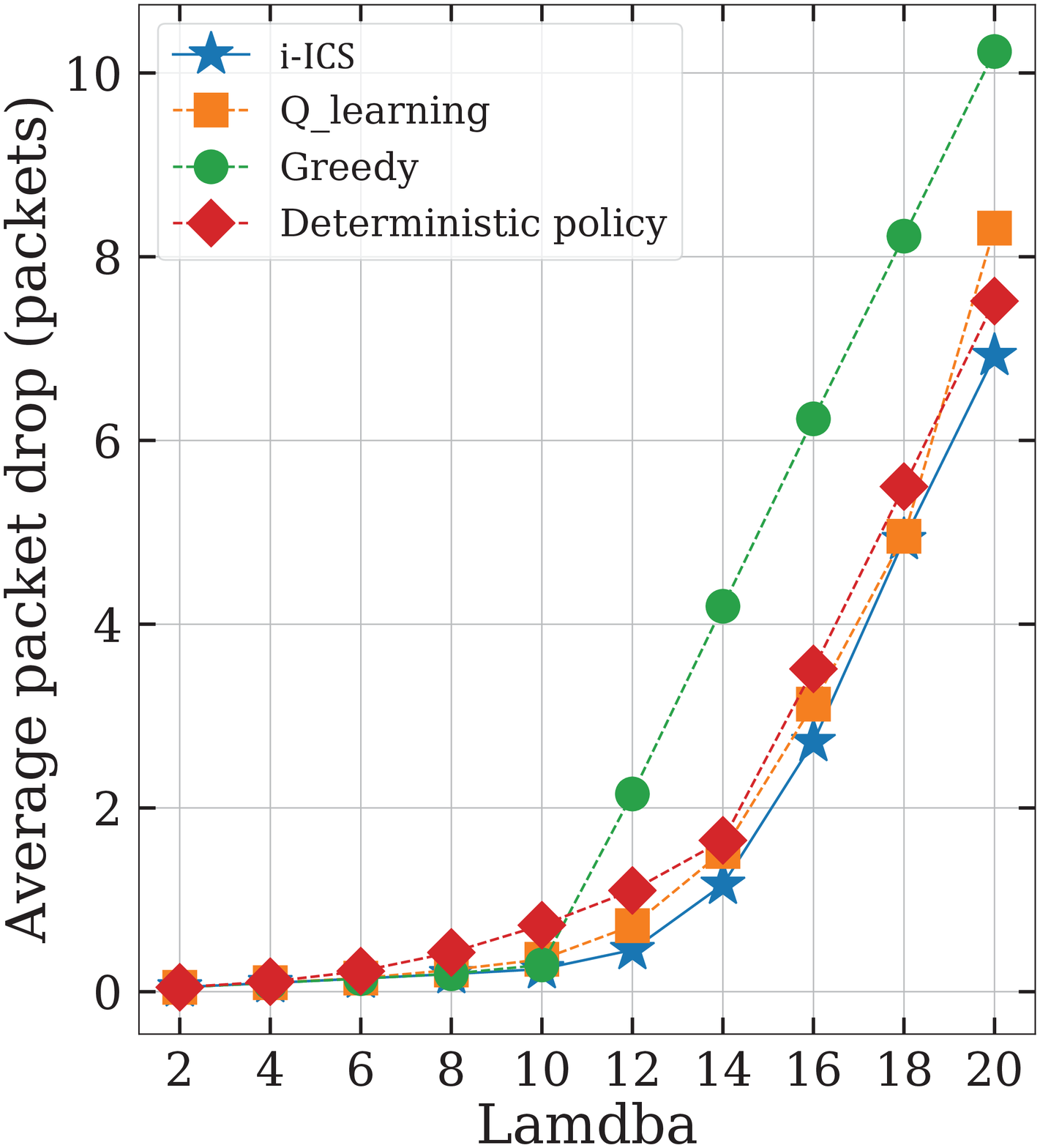}     
			\\
			{\fontsize{10}{1}\text{(a) Average cost}}
			&{\fontsize{10}{1}\text{(b) Average queue length}}
			&{\fontsize{10}{1}\text{(c) Average sensing accuracy}}
			&{\fontsize{10}{1}\text{(d) Average packet drop}}
		\end{array}$
		\caption{Varying data arrival rate with strong channel quality, $w_1\!=\!0.025,w_2\!=\!0.8$, and $w_3\!=\!0.5$.}
		\label{fig:good_w2}
		\vspace{-10pt}
	\end{figure*}
	\subsubsection{Performance Evaluation}
	We then evaluate the robustness of our proposed approach, i.e., i-ICS, by varying the mean number of packets arrived at a time slot $\lambda$ from $2$ to $20$.
	The learned policies of Q-learning and i-ICS  are obtained after $2$$\times$$10^5$ training iterations.
	To evaluate the performance of the considered ICS system, four possible metrics are the average reward, average queue length, average velocity estimation accuracy, and average packet drop.
	Recall that the average reward derived from~\eqref{eq:reward_function} indicates the joint performance of communication and sensing functions.
	Here, to make the demonstration consistent across system performance metrics (i.e., the smaller value, the better system performance), we use a cost metric that is the negation of reward.
	
	We first set the wireless channel to normal quality.
	The weights of the immediate reward function are presented by a weights vectors $\mathbf{W_1}\!=\![0.05,0.4,0.5]$, i.e., $w_1\!=\!0.05$, $w_2\!=\!0.4$, $w_3\!=\!0.5$, as shown in Fig.~\ref{fig:normal_w1}.
	Clearly, the average costs of all policies increase as the packet arrival rate increases, i.e., $\lambda$ increases from 2 to 20, as shown in Fig.~\ref{fig:normal_w1}~(a).
	It stems from the fact that given a data queue with fixed capacity and the ICS system operates under the same environment's characteristics, the higher the value of $\lambda$, the higher the number of dropped packets due to the full packet queue.
	Indeed, Figs.~\ref{fig:normal_w1} (b) and (d) clearly show that when the packet arrival rate increases, the average queue length and average packet drop increase for all policies.
	It can be observed that our proposed algorithm (i.e., i-ICS) achieves the lowest average cost, up to 64.9\% (equivalent to an decrease of 35.1\%) compared with those of other policies.
	Similarly, i-ICS has the lowest average packet drop regardless of the packet arrival rate, and it consistently maintains the average number of packets in the queue as one of the lowest values.
	
	Regarding the sensing metric, i-ICS and Q-learning achieve the highest sensing accuracy (i.e., the lowest value of average velocity estimation accuracy) when $\lambda$ is less than $12$.
	Whereas their average sensing accuracy results are not good compared to other policies if $\lambda$ is larger than $12$.
	The reasons are as follows.
	When the packet arrival rate is low (i.e., $\lambda < 12$), the average numbers of packets in the queue of all policies never pass 30\% of the data queue capacity, as shown in Fig.~\ref{fig:normal_w1}~(b).
	Thus, the ICS-AV can increase the number of frames in the CPI, meaning a decrease in the packet sent in the CPI, to achieve a higher sensing performance without worrying about packet loss due to a full queue.
	Figure~\ref{fig:normal_w1}~(c) clearly shows that the Q-learning and i-ICS can learn this strategy to obtain the best performance in terms of average sensing accuracy when $\lambda < 12$.
		
	As $\lambda$ increases from 12 to $20$, the average queue lengths of the greedy and deterministic policies quickly reach the maximum number of packets that can be stored in the data queue, i.e., $50$, as shown in Fig.~\ref{fig:normal_w1}~(b).
		This can lead to a high possibility of packet drop due to the full data queue.
		Interestingly, sensing performance of Q-learning and i-ICS policies decreases to the worst at $\lambda=16$ and  $\lambda=18$, respectively.
		Then, they manage to increase the sensing accuracy when the packet queue is mostly always full at $\lambda=20$.
		The reason is that the data transmission efficiency is unable to be improved because of a very high packet arrival rate that the system cannot handle.
		Thus, it might be better to improve the sensing accuracy instead of communication efficiency.  
		On the other hand, since the greedy and deterministic policies do not care about the uncertainty of the environment (e.g., the packet drop possibility), they can maintain better sensing accuracy when the packet arrival rate is high.
		However, their transmission efficiency is very low.
		As can be seen in Fig.~\ref{fig:normal_w1}~(d), the average numbers of packet drops of the greedy and deterministic policies are up to 50\% higher than that of our proposed learning algorithm, i.e., i-ICS.	  
		Thus, the proposed algorithm can help the ICS-AV to obtain an optimal policy that strikes a balance between sensing and data transmission metrics, thereby achieving the best overall system's performance compared with those of other policies. 	
		Although i-ICS and Q-learning experience a similar trend when $\lambda$ increases from 2 to 20, i-ICS consistently outperforms Q-learning.
		This stems from the fact that i-ICS can effectively address the high dimensional state in a complicated problem.
		
	Next, we investigate how the immediate reward function's weights can influence the system performance by changing the weight vector to $\mathbf{W_2}\!=\![0.025, 0.8,0.5]$ and varying the packet arrival rate.
	In Fig.~\ref{fig:normal_w2}, it can be observed that the results of deterministic policy are mostly unchanged, except the average cost result, when changing these weights because the ICS-AV's environment is still the same as the previous experiment, and this policy does not rely on the immediate function.
	As the weight of sensing metric (i.e., $w_2$) is doubled, i-ICS, Q-learning, and greedy policies achieve sensing accuracy results that are much better than those in the previous experiment.
	In addition, except for the Q-learning, they also consistently outperform the deterministic policy in terms of sensing metric. 
	In contrast, these policies' data transmission metrics (i.e., the average packet drop and average queue length) become worse than those in the first experiment, as shown in Figs.~\ref{fig:normal_w2} (b) and (d).
	The reason is that when the ratios $w_1/w_2$ and $w_3/w_2$ become smaller, the ICS system pays more attention to the sensing accuracy.
	Thus, Fig.~\ref{fig:normal_w2} clearly shows that in practice, these weights can be adjusted so that our proposed learning algorithm can obtain a policy that fulfils different requirements of a ICS system at different times. 
	Thanks to the ability to learn without requiring complete information of the surrounding environment, i-ICS still achieves the best overall performance when increasing the sensing metric's weight.
	
	We now examine the robustness of our proposed approach by letting it work under different channel qualities, i.e., poor quality with the PER probability vector $\mathbf{p}^p_c\! =\! [0.6,0.2,0.2]$ and good quality with the PER probability vector $\mathbf{p}^g_c\! =\! [0.2,0.2,0.6]$.
	For each of channel qualities, two sets of results are collected according to $\mathbf{W_1}$ and $\mathbf{W_2}$ when varying the packet arrival rate, as shown in Figs.~\ref{fig:bad_w1} to~\ref{fig:good_w2}.
	Overall, all policies' results experience similar trends as those in normal channel quality.
	It can be observed that the channel quality significantly affects the system performance.
	Specifically, the overall system performance (i.e., the average cost) increases as the channel quality changes from poor to normal and then to good regardless of the weight vector.
	The reason is that as the channel quality becomes worst, the probability of packet drop decreases, leading to a better performance of the ICS system. 
	Among the policies, i-ICS achieves the highest overall performance boost when channel quality changes from poor to good, e.g., with $\mathbf{W_1}$, i-ICS's average cost decreases up to 51.7\% while those of the greedy and deterministic policies reduce up to 41.17\%.	
	This is because our proposed approach can effectively adapt its behaviour according to the changes in its surrounding environment to improve the system performance significantly.
	
	\section{Conclusion}
	\label{sec:conclusion}	
In this paper, we have developed a novel MDP-based framework that allows an ICS-AV to automatically and adaptively decide its optimal waveform structure based on the observations to maximize the overall performance of the ICS system.
	Then, we have proposed an advanced learning algorithm, i.e., i-ICS, that can help the ICS-AV gradually learn an optimal policy through interactions with the surrounding environment without requiring complete knowledge about the environment in advance.
	As such, our proposed approach can effectively handle the environment's dynamic and uncertainty as well as the high dimensional state space problem of the underlying MDP framework. 
	The extensive simulation results have clearly shown that the proposed solution can strike a balance between communication efficiency and sensing accuracy, thereby consistently outperforming the benchmark methods in different scenarios.

	
	\ifCLASSOPTIONcaptionsoff
	\newpage
	\fi

\end{document}